\documentclass[%
reprint,
superscriptaddress,
 amsmath,amssymb,
 aps, prx,
longbibliography,
floatfix,
]{revtex4-2}

\usepackage{amsmath, amssymb, amsthm}
\usepackage[utf8]{inputenc}
\usepackage{graphicx}
\usepackage{braket}
\usepackage{bm,bbm}
\usepackage{enumerate}

\newcommand{\abs}[1]{\left\lvert #1 \right\rvert}

\newcommand{\floor}[1]{\left \lfloor #1 \right \rfloor}
\def\ketbra#1#2{{\vert#1\rangle\!\langle#2\vert}}
\newcommand\norm[1]{\left\lVert#1\right\rVert}

\usepackage{newtxtext}
\usepackage{color}
\usepackage{xcolor}
\usepackage[colorlinks=true,
linkcolor={blue!75!black!80!yellow},
citecolor={blue!75!black!80!yellow},
urlcolor={blue!75!black!80!yellow}
]{hyperref}

\theoremstyle{plain}
\newtheorem{theorem}{Theorem}
\newtheorem{lemma}{Lemma}

\theoremstyle{definition}

\newcommand{\chemUofT}{\affiliation{%
    Chemical Physics Theory Group, Department of Chemistry, University of Toronto, Toronto, Ontario, Canada}
    }
\newcommand{\csUofT}{\affiliation{%
    Department of Computer Science, University of Toronto, Toronto, Ontario, Canada}}

\newcommand{\PNNL}{\affiliation{%
    Pacific Northwest National Laboratory, Richland, Washington, USA}}
\newcommand{\vecInst}{\affiliation{%
    Vector Institute for Artificial Intelligence, Toronto, Ontario, Canada}}
    
\newcommand{\chemEngUofT}{\affiliation{%
    Department of Chemical Engineering \& Applied Chemistry, University of Toronto, Toronto, Ontario, Canada}}
\newcommand{\matSciUofT}{\affiliation{%
    Department of Materials Science \& Engineering, University of Toronto, Toronto, Ontario, Canada}}
\newcommand{\CIFAR}{\affiliation{%
    Lebovic Fellow, Canadian Institute for Advanced Research, Toronto, Ontario, Canada}}
\newcommand{\CIFARB}{\affiliation{
Canadian Institute for Advanced Research, Toronto, Ontario, Canada
}}
\newcommand{\AIST}{\affiliation{%
    Research Center for Emerging Computing Technologies, National Institute of Advanced Industrial Science and Technology (AIST), 1-1-1 Umezono, Tsukuba, Ibaraki 305-8568, Japan}}
\newcommand{\Keio}{\affiliation{%
    Quantum Computing Center, Keio University, 3-14-1 Hiyoshi, Kohoku-ku, Yokohama, Kanagawa, 223-8522, Japan}}

\newcommand{\macquarie}{\affiliation{%
Center for Engineered Quantum Systems, School of Mathematical and Physical Sciences, Macquarie University, 2109 NSW, Australia
}}

\newcommand{\IQST}{\affiliation{%
Institute for Quantum Science and Technology,
University of Calgary, Alberta, Canada}}
    
\begin{document}
\title{Fast quantum algorithm for differential equations}
\author{Mohsen Bagherimehrab}
\email{mohsen.bagherimehrab@utoronto.ca}
\chemUofT\csUofT
\author{Kouhei Nakaji}\chemUofT\AIST\Keio
\author{Nathan Wiebe}\csUofT\PNNL\CIFARB
\author{Gavin~K.~Brennen}\macquarie
\author{Barry C.~Sanders}\IQST
\author{Al\'an Aspuru-Guzik}
\chemUofT\csUofT\vecInst\chemEngUofT\matSciUofT\CIFAR

\date{\today}

\begin{abstract}
Partial differential equations (PDEs) are ubiquitous in science and engineering.
Prior quantum algorithms for solving the system of linear algebraic equations obtained from discretizing a PDE have a computational complexity that scales at least linearly with the condition number~$\kappa$ of the matrices involved in the computation.
For many practical applications, $\kappa$ scales polynomially with the size~$N$ of the matrices, rendering a polynomial complexity in~$N$ for these algorithms.
Here we present a quantum algorithm with a complexity that is polylogarithmic in~$N$ but is independent of~$\kappa$ for a large class of PDEs.
Our algorithm generates a quantum state from which features of the solution can be extracted.
Central to our methodology is using a wavelet basis as an auxiliary system of coordinates in which the condition number of associated matrices becomes independent of~$N$ by a simple diagonal preconditioner.
We present numerical simulations showing the effect of the wavelet preconditioner for several differential equations.
Our work could provide a practical way to boost the performance of quantum simulation algorithms where standard methods are used for discretization.
\end{abstract}
\maketitle

Partial differential equations~(PDEs) are an integral part of many mathematical models in science and engineering.
The common approach to solving a PDE on a digital computer is mapping it to a matrix equation by a discretization method and solving the matrix equation using numerical algorithms.
The time complexity of the best-known classical algorithm for solving a linear matrix equation obtained from discretizing a linear PDE scales polynomially with the size~$N$ and condition number~$\kappa$ of matrices involved in the computation~\cite{she94,HHL09}; $\kappa$ is defined as the ratio of the largest to smallest singular values of a matrix~\footnote{The condition number is fundamentally a measure of error sensitivity, and its relation to singular values is a consequence of analyzing the error sensitivity under a specific matrix norm (commonly the 2-norm).
Specifically, the condition number of a function quantifies how sensitive the output of the function is to small changes in the input. In the context of a linear system $Ax=b$, the condition number is the ratio of relative changes in the solution~($\norm{\Delta x}/\norm{x}$) to the relative changes in the input~($\norm{\Delta b}/\norm{b}$).
When measuring the sensitivity under 2-norm, the condition number of a matrix becomes the ratio of its largest to smallest singular values.}.
In contrast, several quantum algorithms, known as quantum linear-system algorithms~(QLSAs)~\cite{HHL09,Amb12,CKS17,GSL+19,LT20,SSD19,CAS+22}, have been developed for solving a system of linear equations with a complexity that grows at least linearly with~$\kappa$, and it cannot be made sublinear in~$\kappa$ for a \textit{general} system of linear equations by standard complexity-theoretic assumptions~\cite{HHL09}.

A key challenge in solving differential equations is the arising ill-conditioned matrices (i.e., matrices with large~$\kappa$) after the discretization.
The condition number of these matrices generally grows polynomially with their size~$N$;
it grows as~$N^2$ for second-order differential equations.
Such ill-conditioned matrices not only lead to numerical instabilities for solutions but also eliminate the exponential improvement of existing quantum algorithms with respect to~$N$.
Despite the extensive previous work on efficient quantum algorithms for differential equations~\cite{CLO21,Ber14,MP16,AKC+19,CL20,kro22}, the challenge of ill-conditioned matrices remains, as these algorithms require resources that grow~with~$\kappa$.

In this paper, we establish a quantum algorithm for a large class of inhomogeneous PDEs with a complexity that is polylogarithmic in~$N$ and is independent of~$\kappa$.
Our algorithm's complexity does not violate the established lower bound on the complexity of solving a linear system~\cite{HHL09} as the lower bound is for generic cases.
With respect to complexity lower bounds, the class of differential equations we consider is similar to the class of `fast-forwardable' Hamiltonians~\cite{GSS21,AA17}:
Hamiltonians that are simulatable in a sublinear time and do not violate the ``no-fast-forwarding" theorem~\cite{BAC+07}.
Therefore, we refer to the class of differential equations solvable in polylogarithmic time as the class of `fast-solvable' differential~equations.

The core of our methodology is utilizing wavelets for preconditioning (i.e., controlling the condition number of) the linear system obtained from discretizing a PDE.
Wavelets form a versatile class of basis functions with appealing features~\cite{Dau92,Bey92,mbm}, such as spatial and frequency localization, making them advantageous for applications in  quantum physics and computation~\cite{BP13,BRS+15,BSB+22,ES16,GSM+22}, and computational chemistry~\cite{HBB+16}.
Notably, wavelets provide an optimal preconditioner for a large class of operators by a simple diagonal matrix~\cite{DK92,dah01,CDD+01,jaf92,bey21}.
The preconditioner is optimal in the sense that the preconditioned matrices have uniformly bounded condition numbers independent of their size;
the condition number is constant and only depends on the wavelet~type (Fig.~\ref{fig:precond}(e)).
The diagonal preconditioner is also structured with a particular pattern on diagonal entries (Fig.~\ref{fig:precond}(b)).
We exploit the optimality and structure of the wavelet preconditioner to perform matrix operations on a quantum computer that leads to a polylogarithmic-time quantum algorithm for solving a broad class of PDEs.

\begin{figure*}
    \centering
    \includegraphics[width=\linewidth]{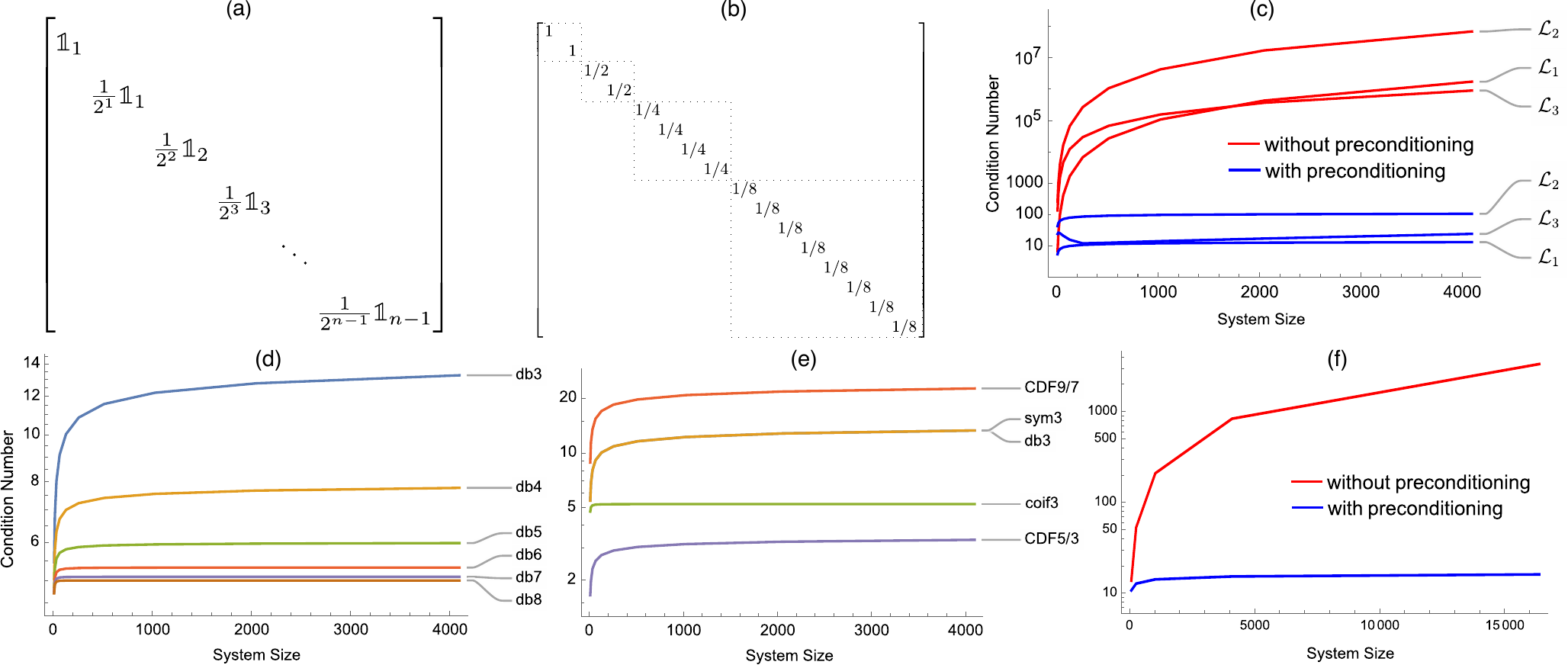}
    \caption{(a)~The diagonal wavelet preconditioner
    with size~$2^n$; $n=4$ is shown in~(b).
    The matrix has a block-diagonal structure with blocks of exponentially increasing size.
    (c)~Condition numbers for differential operators~$\mathcal{L}_1,\mathcal{L}_2$
    and $\mathcal{L}_3$ given in Eq.~\eqref{eq:example_PDEs}, all discretized by finite-difference method on a grid with periodic boundaries.
    Here system size is the number of grid points and the Daubechies wavelet with index~3 is used for preconditioning.
    (d)~The condition number decreases by increasing the wavelet index.
    (e)~The condition number depends on the type of wavelet used but is independent of system size.
    Only~$\mathcal{L}_1$ is preconditioned in~(d) and~(e).
    Symlet~(sym3), Daubechies~(db3) and Coiflet~(coif3) wavelets~\cite{Dau92} with index~3, and
    Cohen–Daubechies–Feauveau~(CDF) wavelets of type ``9/7" and ``5/3"~\cite{CDF92} are used.
    (f)~Condition numbers for 2D Laplacian $\mathcal{L}:=\partial^2_x+\partial^2_y$ discretized by finite-difference method on a square grid with periodic boundaries.}
    \label{fig:precond}
\end{figure*}

We note that the PDE does not need to be discretized using wavelets.
Standard methods such as the finite-difference method can be used for discretization.
In this case, we use wavelets only as a tool to perform linear-algebra manipulations of the linear system obtained by the finite-difference method.
That is to say, wavelets serve as an auxiliary basis in which the condition numbers of matrices involved in the computation are controlled.
This hybrid approach yields a practical way to significantly boost the performance of algorithms based on standard discretization methods and shows the advantage of performing computations in wavelet bases.

Methodologically, we first transform the system of linear equations obtained by the finite-difference discretization into a wavelet basis.
Then we precondition the wavelet-transformed system with the wavelet preconditioner to obtain a linear system with a uniform condition number.
After preconditioning,
a direct approach to obtain a quantum solution for the original linear system is to generate a quantum state encoding the solution of the preconditioned system by QLSAs, which would have a $\kappa$-independent cost, 
and then transform the generated state to a quantum state that encodes the solution of the original system.
This transformation requires applying the nonunitary preconditioner, which needs to be implemented probabilistically.
However, because the condition number of the preconditioner scales linearly with $N=2^n$ (Fig.~\ref{fig:precond}a), the success probability becomes exponentially small in~$n$, leading to an overall $\mathcal{O}(N)$-time algorithm (see Appendix~\ref{apx:direct_approach}).

To achieve polylogarithmic complexity, we avoid directly applying the preconditioner.
Instead, we add two ancilla qubits to the $n$-qubit system and generate a quantum state in the extended $(n+2)$-qubit space that enables computing the expectation value of a given $n$-qubit observable~$M$, a Hermitian operator representing a physical quantity.
To extract this value, we construct an observable~$M'$ in the extended space such that its expectation value with respect to the generated state yields the desired expectation value of~$M$.
Crucially, $M'$ is a $4\times4$ block matrix with~$M$ in each block, a structure we use to construct oracles of~$M'$ by a single query to the oracles of~$M$.

Our approach requires implementing two diagonal unitaries with the same structure as the preconditioner and constructing a block-encoding for the inverse of the preconditioned matrix.
We use the bounded condition number of the preconditioned matrix, the structure of the preconditioner, and properties of a wavelet transformation to perform each operation in polylogarithmic time, which results in a polylogarithmic algorithm for generating a quantum solution for a certain class of PDEs.

\textit{Notation}.---We use $\norm{A}$ for the 2-norm (the operator norm) of an operator $A$, $\norm{\bm{v}}$ for the 2-norm of a vector $\bm{v}$ and $\norm{f}$ for the 2-norm of a function~$f$.
We refer to $A\in \mathbb{C}^{2^n\times 2^n}$ as an $n$-qubit matrix and denote the $n$-qubit identity by $\mathbbm{1}_n$.
We use the technique of block encoding~\cite{GSL+19}, which is a way to embed a matrix as a block of a larger unitary matrix.
Formally, for~$\varepsilon>0$ and a number of qubits~$a$, an $(\alpha, a, \varepsilon)$–block-encoding of an $n$-qubit matrix $A$ is an $(n+a)$-qubit unitary~$U_A$ such that $\norm{A-\alpha(\bra{0^a}\otimes\mathbbm{1}_n)U_A(\ket{0^a}\otimes\mathbbm{1}_n)}\leq \varepsilon$,
where $\alpha\geq \norm{A}$.

\textit{Outline}.---The rest of this paper proceeds as follows.
First, we specify the class of fast-solvable PDEs.
Then we detail the hybrid approach and present the algorithm for PDEs with periodic boundaries, followed by a complexity analysis.
We then extend the algorithm to non-periodic boundaries and describe an alternative approach in which PDEs are directly discretized in a wavelet basis.
We conclude with a discussion and outlook.
Proofs of stated lemmas and detailed complexity analysis are provided in the Appendix.

\textit{Fast-solvable PDEs}.---We first introduce some terminology for PDEs.
A $d$-dimensional~($d$D) linear PDE is formally written as~$\mathcal{L}u(\bm{x})=b(\bm{x})$,
where~$\bm{x}$ is an element of a bounded domain $\Omega \subset \mathbb{R}^d$,
$u(\bm{x})$ and the inhomogeneity~$b(\bm{x})$ are scalar functions, and~$\mathcal{L}$ is a linear operator acting~on~$u(\bm{x})$.
For a linear PDE of order~$m$, $\mathcal{L}$ has the form 
$\sum_{\bm{\alpha}} c_{\bm{\alpha}}(\bm{x}) \partial^{\bm{\alpha}}u$
where $\bm{\alpha}=(\alpha_1,\ldots,\alpha_d)$ is a multi-index with non-negative integers such that $|\bm{\alpha}|:=\alpha_1+\cdots+\alpha_d \leq m$, each coefficient $c_{\bm{\alpha}}(\bm{x})$ is a scalar function, and
$\partial^{\bm{\alpha}}u:=
\partial^{\alpha_1}_1\cdots\partial^{\alpha_d}_d u$ with $\partial^{\alpha_i}_i$ the partial derivative of order $\alpha_i$ with respect to $x_i$.
For all nonzero $\bm{\xi}^{\bm{\alpha}}:=\xi^{\alpha_1}_1\cdots\xi^{\alpha_d}_d$ with $\xi^{\alpha_i}_i \in \mathbb{R}^d$
and all $\bm{x}\in\Omega$,
if $\sum_{|\bm{\alpha}|=m} c_{\bm{\alpha}}(\bm{x})\bm{\xi}^{\bm{\alpha}} \neq0$, then~$\mathcal{L}$ is called `elliptic' operator and its corresponding PDE is called elliptic PDE.

Let $B(u,v):= \int u(\bm{x}) \mathcal{L}v(\bm{x})\,\text{d}\bm{x}$ be the bilinear form induced by~$\mathcal{L}$ on the space~$\mathcal{H}$ of sufficiently differentiable functions;
$B$ is a bilinear map~$\mathcal{H}\times\mathcal{H}\mapsto\mathbb{R}$.
The class of differential operators for which the bilinear form is
\begin{enumerate}[(I)]
    \item\label{cond1}
    symmetric:
    $B(u,v) = B(v,u)$,
    \item\label{cond2}
    bounded: $B(u,v)\leq C \norm{u} \norm{v}$, and
    \item\label{cond3}
    coercive (or elliptic): $B(u,u)\geq c \norm{u}^2$,
    for $0<c< C$,
\end{enumerate}
comprises the PDEs that can be optimally preconditioned by wavelets~\cite{DK92,dah01,CDD+01};
the Lax–Milgram theorem asserts
existence and uniqueness of solution for the variational form of these PDEs~\cite{coh03}.
This class includes prominent examples such as second-order linear and elliptic PDEs with constant or slowly varying coefficients (e.g., the Poisson, Helmholtz, and time-independent
Schr\"{o}dinger equations) as well as higher order equations such as the biharmonic equation~\cite{jaf92,ABC+08}.
A notable subclass is the second-order PDEs known as Sturm-Liouville problems~\footnote{The operator form of Sturm-Liouville problems is
$\mathcal{L}u(\bm{x})=-\sum_{ij=1}^d \partial_j(p_{ij}(\bm{x})\partial_i u)+q(\bm{x})u$,
where $0\leq q(\bm{x})\leq C_q$ and functions $p_{ij}(\bm{x})$ are elements of a $d\times d$ matrix $P_\textsc{sl}$ that satisfy
$c \norm{\bm{v}}^2
\leq \bm{v}^TP_\textsc{sl}\bm{v}
\leq C \norm{\bm{v}}^2 $
for any $\bm{v}\in \mathbb{R}^d$
and positive constants
$c, C$ and $C_q$.
For $d=1$:
$\mathcal{L}u(x) = -\frac{\text{d}}{\text{d}x}
\left(p(x)\frac{\text{d}u}{\text{d}x}\right)
+q(x)u$
with~$c\leq p(x)\leq C$.}.
As illustrative examples, we use three operators defined on the domain $\Omega=[0,1]$
\begin{align}
\label{eq:example_PDEs}
\begin{split}
    \mathcal{L}_1=\frac{\text{d}^2}{\text{d}x^2},
    \quad
    \mathcal{L}_2=\frac{\text{d}^2}{\text{d}x^2}-\frac{\text{d}}{\text{d}x}+\mathbbm{1},\\
    \mathcal{L}_3:=-\frac{\text{d}}{\text{d}x}
    \left(\cosh(x/4)\frac{\text{d}}{\text{d}x}\right)
    +\text{e}^x\mathbbm{1}
\end{split}
\end{align}
in our numerical simulations shown in Fig.~\ref{fig:precond}(c), with $\mathcal{L}_3$ representing an operator with slowly varying coefficients.

\textit{QLSAs}.---Once the PDE $\mathcal{L}u(\bm{x})=b(\bm{x})$ is discretized,
we have a linear system of equations~$A\bm{u}=\bm{b}$ with $\bm{u,b}\in\mathbb{C}^N$ and $A \in\mathbb{C}^{N\times N}$.
We assume $\norm{A}\leq1$ for simplicity;
$A$ and $\bm{b}$ can be rescaled to obey this condition.
We also assume access to a $(1,a,0)$-block-encoding of~$A$ and access to a procedure $\mathcal{P}_{\bm{b}}$
that generates $\ket{\bm{b}}:=\sum_i b_i\ket{i}/\norm{\bm{b}}$, a state that encodes $\bm{b}$ on its amplitudes up to normalization.
Let $\ket{\bm{u}}:=A^{-1}\ket{\bm{b}}=\sum_iu_i\ket{i}$ be the unnormalized state encoding $\bm{u}$ on its amplitudes. 
The quantum approach for `solving' this linear system is to generate a quantum state that enables extracting features of the solution vector by computing the expectation value $\bm{u}^\dagger M\bm{u}$ for a given observable~$M$.
Existing QLSAs generate the state $\ket{\bm{u}}/\norm{\bm{u}}$ that, up to normalization, encodes the solution vector~$\bm{u}$ on its amplitudes.
This state enables extracting the expectation value $\bm{u}^\dagger M\bm{u}=\braket{\bm{u}|M|\bm{u}}$ for a given observable~$M$ using measurement algorithms~\cite{KOS07,ANB+22,HWM+22}.

\textit{The solution state.}---The aim in the quantum approach for solving a linear system is to generate a quantum state (not necessarily $\ket{\bm{u}}$) that enables computing $\braket{\bm{u}|M|\bm{u}}$ for a given~$M$.
Instead of generating $\ket{\bm{u}}$ as conventional approaches, our approach generates another state~$\ket{\psi}$ (specified later) that enables computing the expectation value of interest, i.e., $\braket{\bm{u}|M|\bm{u}}$.
Hereafter, we refer to the state $\ket{\psi}$ in our approach as the `solution state'.
For illustration, first we describe our approach to generate the solution state for ODEs ($d=1$) and then extend it to PDEs.

Our approach involves preconditioning the system~$A\bm{u}=\bm{b}$ in a wavelet basis.
To this end, first we transform the system into a wavelet basis to achieve $A_W\bm{u}_W=\bm{b}_W$, where $\bm{u}_W:=W\bm{u}$, $\bm{b}_W:=W\bm{b}$ and $A_W:=WAW^T$. Here $A_W$ is the wavelet transformation of~$A$, and~$W$ is the $n$-level wavelet transform matrix;
we refer to~\cite[Appendix~A]{BSB+22} and \cite{BA24} for details of the wavelet transform and explicit structure of~$W$.
Preconditioning in the wavelet basis is then achieved by applying the wavelet preconditioner~$P$, resulting in the system
$A_P\bm{u}_P=\bm{b}_P$,
where $\bm{u}_P:=P^{-1}\bm{u}_W$, $\bm{b}_P:=P\bm{b}_W$,
and $A_P:=PA_WP$ is the preconditioned matrix.

\textit{The algorithm}.---By wavelet preconditioning, the state $\ket{\bm{u}}$ can be written as $\ket{\bm{u}} = W^{-1}PA^{-1}_P PW\ket{\bm{b}}$.
We decompose the preconditioner as $P=(U^{+}+U^{-})/2$, where
\begin{equation}
\label{eq:Upm}
    U^\pm:= P \pm\text{i}\sqrt{\mathbbm{1}-P^2}= \text{e}^{\pm\text{i}\arccos{P}}
\end{equation}
is a unitary matrix.
By this decomposition,
$U^{0/1}\equiv U^{+/-}$~and
\begin{equation}
\label{eq:psiabODEs}
\ket{\psi_{ab}}:=W^\dagger U^{a}A_P^{-1}U^{b}W\ket{\bm{b}}
   \quad \forall a,b\in\mathbb{B}:=\{0,1\},
\end{equation}
we obtain
$\ket{\bm{u}}=\frac{1}{4}\sum_{ab}\ket{\psi_{ab}}$, yielding
the identity
\begin{equation}
\label{eq:uMu}
    \braket{\bm{u}|M|\bm{u}} =
    \frac{1}{16}
    \sum_{abcd\in\mathbb{B}} \braket{\psi_{ab}|M|\psi_{cd}}
\end{equation}
for the desired expectation value.
Thus, generating the state
\begin{equation}
\label{eq:psi}
    \ket{\psi}:=
\frac{1}{2\xi}\sum_{ab\in\mathbb{B}}\ket{ab}\ket{\psi_{ab}},
\quad \xi^2:=\frac{1}{4} \sum_{ab\in\mathbb{B}}\norm{\ket{\psi_{ab}}}^2,
\end{equation}
and computing $\braket{\psi|M'|\psi}$ with the observable
\begin{equation}
\label{eq:Mprime}
    M':=\sum_{abcd\in \mathbb{B}}\ket{ab}\bra{cd}\otimes M
\end{equation}
enables computing the desired expectation value as
\begin{equation}
\label{eq:expMprime}
    \braket{\psi|M'|\psi}=
    \frac{1}{4\xi^2}
    \sum_{abcd\in\mathbb{B}} \braket{\psi_{ab}|M|\psi_{cd}} =
    \frac{4}{\xi^2}\braket{\bm{u}|M|\bm{u}}.
\end{equation}

\begin{figure*}
\centering
    \includegraphics[width=\linewidth]{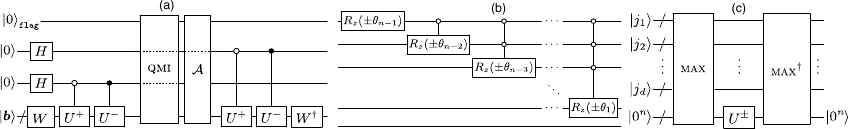}
    \caption{(a)~Quantum circuit for generating the solution state $\ket{\psi}$ in Eq.~\eqref{eq:psi}.
    Multi-qubit gates do not act on dashed-line qubits.
\textsc{qmi} block-encodes the inverse of the preconditioned matrix as in Eq.~\eqref{eq:qmi} and $\mathcal{A}$ denotes amplitude amplification.
    (b)~Implementing~$U^\pm$ in Eq.~\eqref{eq:Udecomp} with classically precomputed angles
    $\theta_1,\ldots \theta_{n-1}$.
    No action is performed on the last qubit.
    (c)~Implementing $U^\pm_{d\text{D}}$ in Eq.~\eqref{eq:UdD} with \textsc{max} defined in~Eq.~\eqref{eq:max}.
     }
    \label{fig:qCircs}
\end{figure*}

We thus need to construct a circuit that generates the solution state $\ket{\psi}$ and
construct oracles of $M'$ from oracles of~$M$.
As~$M'$ is a $4\times4$ block matrix of blocks $M$, oracles of $M'$ can be constructed using two ancilla qubits and one query to oracles of~$M$ (see Appendix~\ref{app:oraclesM}).
The circuit in Fig.~\ref{fig:qCircs}(a) generates $\ket{\psi}$ based on~$\ket{\psi_{ab}}$ defined in Eq.~\eqref{eq:psiabODEs}.
The operation \textsc{qmi} in this circuit is an
$(\alpha,a,\varepsilon)$-block-encoding for the inverse of the preconditioned matrix $A_P$.
Specifically, for $\varepsilon>0$, \textsc{qmi} is a unitary such that
\begin{equation}
\label{eq:qmi}
    \norm{A^{-1}_P-\alpha(\bra{0^a}\otimes\mathbbm{1}_n)\textsc{qmi}
    (\ket{0^a}\otimes\mathbbm{1}_n)}
    \leq \varepsilon,
\end{equation}
where $a$ is the number of ancilla qubits used in block-encoding and $\alpha\geq \lVert A^{-1}_P\rVert$.
For simplicity, only one of the ancilla qubits is shown in Fig.~\ref{fig:qCircs}(a); this qubit serves as a flag to indicate the success of our algorithm.
We provide a procedure for executing \textsc{qmi} in Appendix~\ref{apx:qmi}.

We now describe how to implement~$U^\pm$.
By Eq.~\eqref{eq:Upm}, $U^\pm$ is a diagonal matrix with the same block-diagonal structure as~$P$ in Fig.~\ref{fig:precond}(a):
diagonals in each block of~$U^\pm$ are a constant of the form~$\text{exp}(\pm\text{i}\theta)$ for some $\theta$.
Let $\Lambda_0^r(R_z(\theta))$ denote the $\ket{0^r}$-controlled-$R_z(\theta)$ with $R_z(\theta):=\exp{(\text{i}\theta Z)}$; then we have
\begin{equation}
\label{eq:Udecomp}
    U^\pm = \prod_{r=1}^{n-1}\Lambda_0^{r-1}
    (R_z(\pm\theta_{n-r}))\otimes \mathbbm{1}_{n-r},
\end{equation}
where $\theta_r:=\arccos(1/2^r)$ and~$\mathbbm{1}_n$ is the $n$-qubit identity.
This decomposition yields the circuit in
Fig.~\ref{fig:qCircs}(b)
for~$U^\pm$.

\textit{Extension to PDEs}.---Our approach can be extended to $d$D PDEs by constructing a $d$D preconditioner
and performing $d$D wavelet transformation.
The latter is achieved by the tensor product of 1D wavelet transformations.
We construct the $d$D preconditioner in
Appendix~\ref{apx:PdD} and state, in Lemma~\ref{lemma:PdD}, its action on basis states
$\ket{j_1}\ket{j_2}\cdots\ket{j_d}$,
the tensor product of $n$-qubit states. 
For $d=1$, the 1D preconditioner in Fig.~\ref{fig:precond}(a) acts on basis state $\ket{j}$ as $P\ket{j}=2^{-\floor{\log_2 j}} \ket{j}$
with no action on $\ket{0}$.

\begin{lemma}
\label{lemma:PdD}
The action of $d$\textup{D} preconditioner on basis states is
\begin{equation}
    P_{d\textup{D}} \ket{j_1}\ket{j_2}\cdots\ket{j_d}
    = 2^{-\floor{\log_2 j_{\max}}} \ket{j_1}\ket{j_2}\cdots\ket{j_d},
\end{equation}
where $j_{\max}:=\max(j_1,\ldots,j_d)$;
action is trivial if $j_{\max}=0$.
\end{lemma}

Note that $P_{d\textup{D}}$ is a diagonal matrix whose diagonals have a value~$\leq1$. 
Hence, similar to 1D case, we decompose the $d$D preconditioner as $P_{d\text{D}}=\frac{1}{2}\left(U^+_{d\text{D}}+U^{-}_{d\text{D}}\right)$, a linear combination of two unitaries.
As in the 1D case, here we only need an implementation for $U^\pm_{d\text{D}}$ to generate the solution state for PDEs.
The solution state is $\ket{\psi}$ in Eq.~\eqref{eq:psi} with
\begin{equation}
\label{eq:psiabPDEs}
\ket{\psi_{ab}}=W_{d\text{D}}^\dagger U_{d\text{D}}^{a}A_P^{-1}U_{d\text{D}}^{b}W_{d\text{D}}\ket{\bm{b}}
\quad \forall a,b\in \{0,1\},
\end{equation}
where $W_{d\text{D}}$ is the $d$D wavelet preconditioner.

To construct a circuit for $U^\pm_{d\text{D}}$, let us revisit the 1D case~$U^\pm$.
By the action of 1D preconditioner on basis states, we have
$U^\pm\ket{j}=\text{e}^{\pm\text{i}\theta_j}\ket{j}$, where $\cos\theta_j=2^{-\floor{\log_2 j}}$.
This relation yields $n-1$ distinct nonzero angles, enabling the compilation in Fig.~\ref{fig:qCircs}(b) for~$U^\pm$.
Likewise, the action of $U^\pm_{d\text{D}}$ on basis states is obtained from the preconditioner's action as
\begin{equation}
\label{eq:UdD}
    U^\pm_{d\text{D}} \ket{j_1}\ket{j_2}\cdots\ket{j_d}
    = \text{e}^{\pm\text{i}\theta_{\max}} \ket{j_1}\ket{j_2}\cdots\ket{j_d},
\end{equation}
where $\cos{\theta_{\max}} = 2^{-\floor{\log_2 j_{\max}}}$.
To implement $U^\pm_{d\text{D}}$, first we compute $j_{\max}$ into an $n$-qubit ancilla register by the operation
\begin{equation}
\label{eq:max}
    \textsc{max}
    \ket{j_1}\ket{j_2}\cdots\ket{j_d} \ket{0^n}
    :=
    \ket{j_1}\ket{j_2}\cdots\ket{j_d} \ket{j_{\max}},
\end{equation}
and then apply $U^\pm$ on the ancilla register followed by uncomputing this register;
see the circuit in Fig.~\ref{fig:qCircs}(c).

Having an implementation for $U^\pm_{d\text{D}}$, the circuit for generating the solution state for PDEs becomes similar to that of ODEs shown in Fig.~\ref{fig:qCircs}(a).
Specifically, the $n$-qubit register in Fig.~\ref{fig:qCircs}(a) encoding~$\ket{\bm{b}}$ is replaced with a $dn$-qubit register encoding $\ket{\bm{b}}$ for a PDE,
the QWT $W$ is replaced with $d$ parallel QWTs and $U^\pm$ is replaced with~$U^\pm_{d\text{D}}$.
The block-encoding \textsc{qmi} and amplitude amplification~$\mathcal{A}$ in the circuit remain operationally unchanged; \textsc{qmi} now block encodes the inverse of the preconditioned matrix associated with a PDE.

\textit{Complexity}.---We now analyze our algorithm's complexity.
The algorithm generates the state $\ket{\psi}$ in Eq.~\eqref{eq:psi} with $\ket{\psi_{ab}}$ in Eq.~\eqref{eq:psiabODEs} for ODEs and in Eq.~\eqref{eq:psiabPDEs} for~PDEs.
First we analyze the complexity of generating $\ket{\psi}$ for ODEs.
As per Fig.~\ref{fig:qCircs}~(a), generating this state requires performing four main operations:
$W$ and its inverse;
controlled-$U^\pm$;
block-encoding \textsc{qmi};
and amplitude amplification $\mathcal{A}$.
The QWT~$W$ on $n$ qubits can be executed by $\mathcal{O}(n^2)$ gates~\cite{BA24,FW99}.
By Eq.~\eqref{eq:Udecomp},
$U^\pm$ is a product of multi-control rotations, and so is the controlled-$U^\pm$ but with an extra control.
This observation yields the following lemma.

\begin{lemma}
\label{lemma:cUpm}
    controlled-$U^\pm$ with $U^\pm$ in Eq.~\eqref{eq:Udecomp} can be executed by $\mathcal{O}(n)$ Toffoli, one- and two-qubit gates, and $n$ ancilla qubits.
\end{lemma}

The block-encoding \textsc{qmi} can be constructed using existent QLSAs~\cite{CKS17,GSL+19} by $\mathcal{O}(\kappa_p\log(\kappa_p/\varepsilon))$ calls to a block-encoding of $A_P$;
see Appendix~\ref{apx:detailedComplexity}.
We use $A_P= PWAW^\dagger P$ to construct a block-encoding of $A_P$ by one use of the block-encoding of~$A$, and $\mathcal{O}(n^2)$ gates (Appendix~\ref{apx:BEncode_Ap}).
This block-encoding and $\kappa_p\in \mathcal{O}(1)$ in our application yield Lemma~\ref{lemma:qmi}.
We remark that the bounded condition number can also be used to improve the complexity of existent QLSAs (see Appendix~\ref{apx:qlsa}).

\begin{lemma}
\label{lemma:qmi}
The \textsc{qmi} in Eq.~\eqref{eq:qmi} can be executed by $\mathcal{O}(\log(1/\varepsilon))$ uses of the block-encoding of~$A$ and $\mathcal{O}(n^2)$ gates.
\end{lemma}

We now analyze the cost of amplitude amplification.
Let~$\ket{\Psi}$ be the input state to~$\mathcal{A}$ in Fig.~\ref{fig:qCircs}(a).
This state can be written as
$\ket{\Psi}=
\sqrt{p}\ket{0}_\texttt{flag}\!\ket{\text{G}}+
\sqrt{1-p}\ket{\text{B}}$, where $\ket{0}_\texttt{flag}\!\ket{\text{G}}$ is the `good' part of~$\ket{\Psi}$ whose amplitude is to be amplified and~$\ket{\text{B}}$ is the `bad' part with
$({}_\texttt{flag}\!\bra{0}\otimes\mathbbm{1}_{n+1})\ket{\text{B}}=0$.
The success probability~$p$ is $\mathcal{O}(1)$ because \textsc{qmi} block encodes the inverse of a matrix with $\kappa\in\mathcal{O}(1)$.
Although $p$ is constant, its value is unknown.
Consequently, only~$\mathcal{O}(1)$ rounds of amplifications are needed on average to boost $p$ to say $2/3$~\cite[Theorem~3]{BHM+02}.
Each round can be implemented by a reflection about $\ket{\Psi}$ and a reflection about~$\ket{0}_\texttt{flag}$~\cite{BHM+02}.
The latter reflection is independent of~$n$, and the former is achieved by performing the inverse of operations that generate $\ket{\Psi}$ from the all-zero state $\ket{0^{n+2}}$, reflecting about it, and then performing the state-generation operations.
Given access to a procedure~$\mathcal{P}_{\bm{b}}$ that generates $\ket{\bm{b}}$, $\ket{\Psi}$ is generated using $\mathcal{P}_{\bm{b}}$,~$H$, $W$, controlled-$U^\pm$ and \textsc{qmi}.
Lemma~\ref{lemma:AA} follows from the above analysis;
the proof and detailed implementation of~$\mathcal{A}$ is given in Appendix~\ref{apx:detailedComplexity}.
Altogether, the overall cost to generate the solution state for ODEs scales as $\mathcal{O}(n^2)$ with $n$, which is polylogarithmic in~$N=2^n$.

\begin{lemma}
\label{lemma:AA}
    The amplitude amplification $\mathcal{A}$ in Fig.~\ref{fig:qCircs}(a) can be executed by $\mathcal{O}(1)$ uses of $\mathcal{P}_{\bm{b}}$, $\mathcal{O}(\log(1/\varepsilon))$ uses of the block-encoding of~$A$ and $\mathcal{O}(n^2)$ gates, all on average.
\end{lemma}

Note that the overall cost to generate the solution state for ODEs is determined by the cost of $W$ and controlled-$U^\pm$.
For PDEs, these operations are replaced with their $d$D versions: $W_{d\text{D}}$ and controlled-$U_{d\text{D}}^\pm$.
Similar to ODEs, these operations determine the overall cost to generate the solution state for PDEs.
The cost of $W_{d\text{D}}=\otimes_{i=1}^d W$ is $d$ times the cost of~$W$.
By Fig.~\ref{fig:qCircs}(c), the cost of controlled-$U_{d\text{D}}^\pm$ is obtained by adding the cost of controlled-$U^\pm$ in Lemma~\ref{lemma:cUpm} and \textsc{max} in Lemma~\ref{lemma:max}.

\begin{lemma}
\label{lemma:max}
    The operation \textsc{max} in Eq.~\eqref{eq:max} can be executed using $\mathcal{O}(dn)$ Toffoli gates and $\mathcal{O}(n)$ ancillary qubits.
\end{lemma}

Altogether, we have the following theorem for the cost of generating a quantum solution for $d$D PDEs; $d=1$ for ODEs.

\begin{theorem}
\label{theorem}
Let $\mathcal{L}u(\bm{x})=b(\bm{x})$ be a $d$\textup{D} inhomogeneous linear PDE on the domain $[0,1]^d$ with periodic boundaries,
where the bilinear form of $\mathcal{L}$ is symmetric, bounded and elliptic;
see conditions~(\ref{cond1}--\ref{cond3}).
Let $A\bm{u}=\bm{b}$ with
$A \in\mathbb{R}^{N\times N}$
and $\bm{u,b}\in \mathbb{R}^N$ be the linear system obtained by the finite-difference method on a $d$\textup{D} grid with $N=2^{nd}$ points.
For~$\varepsilon>0$ and given access to a $(1,a,0)$-block-encoding $U_A$ of~$A$ and a procedure~$\mathcal{P}_{\bm{b}}$ that generates $\ket{\bm{b}}$, an $\varepsilon$-approximation of the solution state~$\ket{\psi}$ in Eq.~\eqref{eq:psi} can be generated by $\mathcal{O}(1)$ uses of $\mathcal{P}_{\bm{b}}$, $\mathcal{O}(\log(1/\varepsilon))$ uses of $U_A$ and $\mathcal{O}(dn^2)$ gates, all on average.
\end{theorem}

\textit{Non-periodic boundaries}.---The periodic boundary condition considered so far allows using the standard wavelet transformation for preconditioning. For non-periodic boundaries, wavelets defined on bounded domains with appropriate boundary wavelets can be used to enforce the desired boundary condition, which would need incomplete wavelet transforms~\cite{CDV93,NHT+96,HM18}. However, a simple way to apply the proposed algorithm for PDEs with non-periodic boundaries is to use the method of images to construct a suitable periodic extension of the problem in an extended domain. Solving the PDE with periodic boundaries in the extended domain then yields the correct solution on the original domain (see Appendix~\ref{apx:periodic_extension} for details).

\textit{The wavelet approach}.---In the hybrid approach described above, the finite-difference method is used for discretization and wavelets serve as an auxiliary basis for preconditioning.
An alternative approach, which we call it the wavelet approach, is to discretize the PDE directly in a wavelet basis.
In this case, the differential operator~$\mathcal{L}$ and the inhomogeneity~$b(\bm{x})$ are already represented in the wavelet basis, so the QWT and its inverse in Fig.~\ref{fig:qCircs} are not needed.
As a result, the gate cost in Theorem~\ref{theorem} is reduced from~$\mathcal{O}(n^2)$ to~$\mathcal{O}(n)$, as the~$\mathcal{O}(n^2)$ cost is due to the QWT~(Appendix~\ref{apx:detailedComplexity}).
However, the wavelet approach requires that the block-encoding~$U_A$ and the procedure $\mathcal{P}_{\bm{b}}$ correspond to $A$ and $\bm{b}$ represented in the wavelet basis.
As in the hybrid approach, the condition number of the preconditioned matrix in the wavelet approach is effectively constant.
Fig.~\ref{fig:comparison} compares the condition numbers in both approaches.

\begin{figure*}
    \centering
    \includegraphics[width=\linewidth]{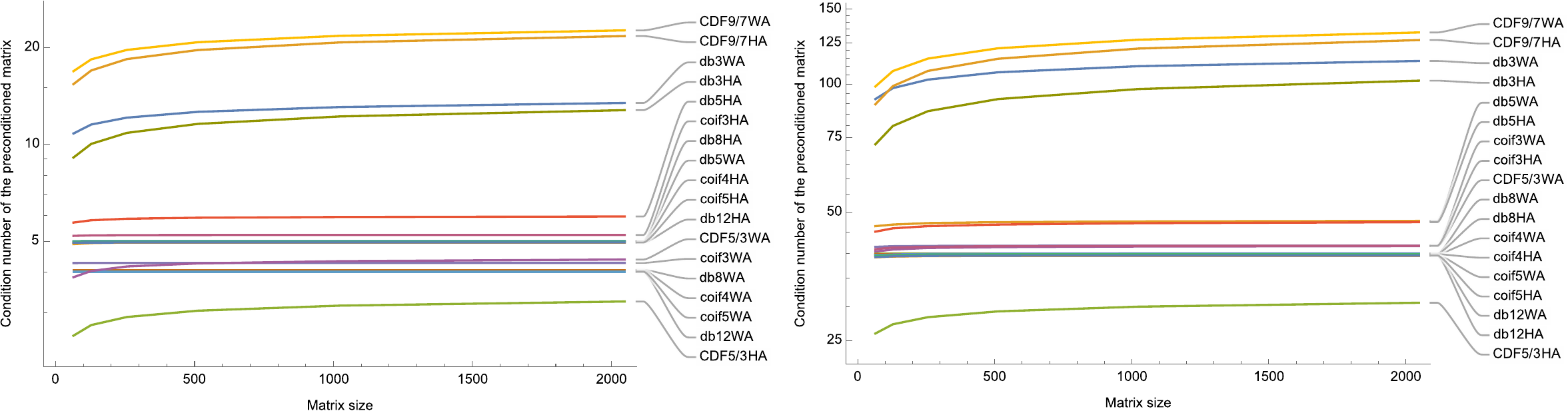}
    \caption{Condition numbers of the preconditioned matrix in the hybrid approach~(HA) vs the wavelet approach~(WA)
    for $\mathcal{L}_1$~(Left) and $\mathcal{L}_2$~(Right), given in Eq.~\eqref{eq:example_PDEs}, using various wavelets.
    The CDF wavelet of type ``5/3'' with the hybrid approach yields the lowest condition~number.
    }
    \label{fig:comparison}
\end{figure*}

\textit{Conclusion and discussion}.---We have presented a quantum algorithm for a broad class of inhomogeneous PDEs, with a runtime that is independent of the condition number of matrices involved in the computation.
Specifically, the algorithm runs in polylogarithmic time with respect to the matrix dimension and prepares a quantum state from which features of the solution vector can be efficiently extracted, thereby achieving
an exponential speedup over classical methods.
The algorithm applies to the class of PDEs for which wavelets provide an optimal preconditioner.
This class includes prominent examples such as Sturm-Liouville problems and second-order linear and elliptic PDEs with constant or slowly varying coefficients.
As a result, the algorithm provides a promising pathway toward realizing quantum advantage in practical applications.

The hybrid approach we used demonstrates that wavelets can serve as an auxiliary basis in which the condition numbers of matrices obtained from standard finite-difference discretization are effectively constant, thereby substantially reducing the cost of quantum matrix operations.
This approach could have a broad impact on quantum algorithms research.
In particular, it could provide a practical way to boost the performance of quantum simulation algorithms, such as those used in quantum chemistry~\cite{KJL+08,MFZ+22}, where standard discretization methods are used.
Our approach also provides a practical way to accurately estimate the condition number of large-size systems, which is often assumed as input to quantum algorithms~\cite{HHL09,CKS17}.

We note that our algorithm makes no assumptions about the inhomogeneity $b(\bm{x})$,
in contrast to those in Refs.~\cite{TAW+21,ALW+22}, which exploit specific structure in $b(\bm{x})$ to achieve a cost independent of the condition number.
We also note that although the condition number becomes constant by preconditioning, its value may still be large.
Choosing a suitable wavelet, as illustrated in Fig.~\ref{fig:precond}, or tailoring the diagonal preconditioner to the specific problem, as in Refs.~\cite{CDD+01,bey21}, can further decrease the condition number. Instantiating the proposed algorithm for concrete applications by constructing the input oracles and performing full implementations remains an important direction for future research.

\textit{Acknowledgement}.---MB thanks Joshua T. Cantin and Philipp Schleich for helpful comments.
MB and AAG acknowledge the generous support and funding of this project by the Defense Advanced Research Projects Agency (DARPA) under Contract No. HR0011-23-3-0021. Any opinions, findings and conclusions or recommendations expressed in this material are those of the author(s) and do not necessarily reflect the views of the Defense Advanced Research Projects Agency.
AAG also acknowledges support from the Canada 150 Research Chairs program and NSERC-IRC. AAG also acknowledges the generous support of Anders G. Frøseth.
NW acknowledges funding for this work from the US DOE National Quantum Information Science Research Centers, Co-design Center for Quantum Advantage (C2QA) under contract number DE-SC0012704.
G.K.B. acknowledges support from the Australian Research Council Centre of Excellence for Engineered Quantum Systems (Grant No. CE 170100009).


\bibliography{references}

\onecolumngrid
\appendix

\section{Oracles for the observable in the extended space}
\label{app:oraclesM}

In this appendix, we construct oracles that specify the observable $M'$ in the extended space, given in Eq.~\eqref{eq:Mprime}, using the oracles specifying the observable $M$.
Let
\begin{align}
    &O^{\text{loc}}_M: \ket{j}_\texttt{row} \ket{\ell}_\texttt{indx} \ket{0}_\texttt{col}
    \mapsto \ket{j}_\texttt{row} \ket{\ell}_\texttt{indx} \ket{\text{loc}(j,\ell)}_\texttt{col}, \\
    &O^{\text{val}}_M: \ket{j}_\texttt{row} \ket{k}_\texttt{col} \ket{0}_\texttt{val}
    \mapsto \ket{j}_\texttt{row} \ket{k}_\texttt{col} \ket{M_{jk}}_\texttt{val},
\end{align}
be the oracles that specify the location and value of nonzero entries of $M$,
where the function $\text{loc}(j,\ell)$ returns the column index of the $\ell$th nonzero in the $j$th row of $M$.
Then the oracles specifying $M'$ can be constructed using a single query to the above oracles~as
\begin{align}
    &O^{\text{loc}}_{M'}: \ket{a,b,j}_\texttt{row} \ket{c,d,\ell}_\texttt{indx} 
    \ket{0}_\texttt{col} 
    \mapsto \ket{a,b,j}_\texttt{row}  \ket{c,d,\ell}_\texttt{indx}  \ket{(2c+d)N+\text{loc}(j,\ell)}_\texttt{col} , \\
    &O^{\text{val}}_{M'}: \ket{a,b,j}_\texttt{row}  \ket{c,d,k}_\texttt{col}  \ket{0}_\texttt{val} 
    \mapsto \ket{a,b,j}_\texttt{row}  \ket{c,d,k}_\texttt{col}  \ket{ M_{jk}}_\texttt{val},
\end{align}
where the binary numbers $a,b,c$ and $d$ encode the location of the blocks in the block matrix $M'$ in Eq.~\eqref{eq:Mprime}.

\section{Direct approach for generating the solution state}
\label{apx:direct_approach}

In this section, we show that the direct approach for generating the state $\ket{\bm{u}}=A^{-1}\ket{\bm{b}}$ has a computational cost that scales at least linearly with the size $N$ of the matrix~$A$.
This state can be written as
$\ket{\bm{u}} = W^{\dagger} P A_P^{-1} P W \ket{\bm{b}}$.
In the direct approach, this state is generated as follows.
First perform the quantum wavelet transformation $W$ on $\ket{\bm{b}}$
to obtain $\ket{\bm{b}_W} := W\ket{\bm{b}}$.
Then apply~$P$ on this state to generate the state
$\ket{\bm{b}_P} := P \ket{\bm{b}_W}$ and afterward apply $A^{-1}_P$ on $\ket{\bm{b}_P}$ to obtain $\ket{\bm{u}_W} := P \ket{\bm{u}_P}$.
Finally, perform the inverse of $W$
to have $\ket{\bm{u}} = W^{\dagger}\ket{\bm{u}_W}$.

The preconditioner $P$ and the inverse of the preconditioned matrix $A^{-1}_P$ are not unitary operations, so these operations need to be implemented by a unitary block encoding.
Because the condition number $\kappa_p$ of~$A_P$ is bounded by a constant number, the success probability of the block-encoding of $A^{-1}_P$ is $\Omega(1)$; see Appendix~\ref{apx:qmi} and~\ref{apx:qaa}.
The constant success probability enables efficient implementation for~$A^{-1}_P$.
In contrast, the condition number of the preconditioner~$P$ scales as $\Theta(N)$, causing the success probability of its probabilistic implementation to scale as~$\Omega(1/N^2)$, as we show in the following.
Let $U_P$ be the $(1,1,0)$-block-encoding of $P$ in Eq.~\eqref{eq:U_P}, then we have
\begin{equation}
    U_P\ket{0} \ket{\bm{b}_W} = \ket{0} P\ket{\bm{b}_W} + \ket{\phi},
\end{equation}
where $\ket{\phi}$ is a state such that $(\bra{0}\otimes\mathbbm{1}_n)\ket{\phi}=0$.
If we measure the single qubit in the computational basis and obtain $\ket{0}$, then a normalized version of the state $P\ket{\bm{b}_W}=\ket{\bm{b}_P}$ is prepared on the second register.
The success probability $p_\text{succ}$ to obtain $\ket{0}$ as a result of the measurement is
\begin{equation}
    p_\text{succ} = \norm{P\ket{\bm{b}_W}}^2
    \geq \min_{\ket{\bm{b}_W}} \norm{P\ket{\bm{b}_W}}^2
    = \lambda^2_{\min}(P) = (2/N)^2,
\end{equation}
where $\lambda_{\min}(P)$ is the smallest eigenvalue of $P$.
Here we used $\lambda_{\min}(P)=2/N$ as per Fig.~\ref{fig:precond}(a).
We also have $\lambda_{\max}(P)=1$, so the condition number of $P$ scales as $\Theta(N)$.
To boost the success probability to $\Omega(1)$, we need $1/\sqrt{p_\text{succ}}$ rounds of amplitude amplifications.
This number of amplifications makes the computational cost of the direct approach for generating $\ket{\bm{u}}$ at least linear in $N$.

\section{Periodic extension using the method of images}
\label{apx:periodic_extension}

Here we describe how the proposed quantum algorithm for periodic boundary condition can be adapted to the Neumann and Dirichlet boundary conditions using the method of images.
For simplicity, we consider one-dimensional PDEs on the unit interval $\Omega=[0,1]$ and assume homogeneous boundary conditions, i.e., $u(0)=u(1)=0$ for Dirichlet boundary conditions and $u'(0)=u'(1)=0$ for Neumann boundary conditions.
We further assume the inhomogeneity function $b(x)$ and the coefficient functions, in the case of variable-coefficient PDEs, are smooth across the boundary.
The periodic extension using the method of images in quantum algorithms for PDEs was explored in Refs.~\cite{CLO21,KAL+25}.
Our description here is included for completeness and to provide further details.

The method of images embeds the functions involved in a PDE (such as the solution function, the inhomogeneity, and the coefficient functions for PDEs with variable coefficients) into an extended domain with carefully chosen symmetry to enforce the desired boundary conditions.
Specifically, each function originally defined on $[0,1]$ is extended to $[-1,1]$ using either an even or odd extension, depending on the type of boundary condition (Dirichlet or Neumann).
This extended function is then periodically extended to all of $\mathbb{R}$.
The result is that the extended functions satisfy the periodic boundary condition and the differential operator becomes periodic as well.
That is, the finite-difference representation of the differential operator becomes a matrix with wrap-around terms at the boundaries.

More precisely, let $f(x)$ be any of the functions in the PDE defined on the domain $[0,1]$.
Its odd $\tilde{f}_-$ and even $\tilde{f}_+$ extensions to the domain~$[-1,1]$ are given by
\begin{equation}
    \tilde{f}_\pm(x) =
    \begin{cases}
        f(x) & x \in [0,1],\\
        \pm f(-x) & x\in [-1,0).
    \end{cases}
    \nonumber
\end{equation}
The odd extension is used for the Dirichlet boundary condition.
This extension makes the solution $\tilde{u}(x)$ antisymmetric about $x=0$, so it vanishes at the boundaries of the original domain, i.e., $\tilde{u}(0)=\tilde{u}(1)=0$, which satisfies the Dirichlet boundary condition.
For the Neumann boundary condition, the even extension is used. This makes the solution function symmetric about~$x=0$, so its derivative vanishes at boundaries of the original domain (i.e., $\tilde{u}'(0)=\tilde{u}'(1)=0$), matching the Neumann boundary condition.

Once the odd/even extension is constructed on $[-1,1]$, it is extended to all of $\mathbb{R}$ by $\tilde{f}(x+2)=\tilde{f}(x)$.
The result is a periodic and piecewise smooth function that satisfies the periodic boundary condition.
The PDE $\tilde{\mathcal{L}}\tilde{u}
(x)=\tilde{b}(x)$ on the extended domain then becomes equivalent to the PDE $\mathcal{L}u(x)=b(x)$ on the original domain, in the sense that the solution on the extended domain, when restricted to~$[0,1]$, recovers the solution to the original PDE.

If $N=2^n$ grid points are used for discretizing the original PDE on the unit interval, we now use $2N$ grid points over~$[-1,1]$ for discretizing the extended PDE. We thus need an extra qubit to encode the extended problem using quantum states. Let the vector $\bm{f} = (f_0,f_1,\ldots,f_{N-1})$ be the discretized version of a function $f(x)$ on the original domain. Then the discretized version of the extended function $\tilde{f}(x)$ is
$\tilde{\bm{f}}_\pm = (\pm f_{N-1},\ldots,\pm f_0, f_0,\ldots,f_{N-1})$ with the minus sign for the odd extension (Dirichlet boundary) and the plus sign for the even extension (Neumann boundary).

To apply the proposed algorithm to the periodized PDE in the extended domain, we construct a procedure $\mathcal{P}_{\tilde{\bm{b}}}$ for preparing the state $\ket{\tilde{\bm{b}}}:=\sum_{j=0}^{2N-1} \tilde{b}_j\ket{j}/\tilde{\norm{\bm{b}}}$ that encodes the vector $\tilde{\bm{b}}$, the discretized version of the extended inhomogeneity function~$\tilde{b}(x)$, on the amplitudes using the given procedure $\mathcal{P}_{\bm{b}}$ that prepares the state $\ket{\bm{b}}=\sum_{j=0}^{N-1}b_j\ket{j}/\norm{\bm{b}}$.
This is achieved simply by implementing the shift operation $\textsc{shift}: \ket{j} \mapsto \ket{(j+N-1) \bmod N}$ and simple quantum gates. 
Specifically, we have
\begin{equation}
    \ket{\tilde{\bm{b}}} =
    \frac1{\sqrt{2}}
    (\ket{1}\ket{\bm{b}}
    \pm\ket{0}\textsc{shift}\ket{\bm{b}}),
\end{equation}
where the plus sign is for the even extension and minus sign is for the odd extension.
The state with the plus sign is prepared by applying a Hadamard gate on the first qubit, $\mathcal{P}_{\bm{b}}$ on the remaining $n$ qubits followed by $\ket{0}$-controlled \textsc{shift} operation.
The state with the minus sign is achieved by applying a $Z$ gate followed by an $X$ gate before the controlled shift operation.

\section{Wavelet preconditioner for PDEs}
\label{apx:PdD}

In this appendix, we describe a way to construct the multidimensional diagonal wavelet preconditioner and obtain its action on basis states given in Lemma~\ref{lemma:PdD}.
We refer to~\cite{DK92,dah01,CDD+01} and~\cite[\S3.11]{coh03} for rigorous construction.
We also present a procedure for implementing the multidimensional preconditioner.

\subsection{Multidimensional diagonal preconditioner}
\label{apx:multiD}

Inspired by~\cite[\S3.1]{ABC+08}, we construct the $d$-dimensional diagonal preconditioner by analogy with the one-dimensional case.
The 1D preconditioner is constructed from the multi-scale property of a wavelet basis: the basis functions have an associated scale index~$s$, and the support of these functions is proportional to $2^{-s}$.
At the base scale $s=0$ and at higher scales $s>0$.
The number of basis functions at scale $s$ is two times the number of basis functions at scale $s-1$.
The 1D preconditioner in Fig.~\ref{fig:precond}(a) follows from these properties of a wavelet basis: each diagonal entry has a form $2^{-s}$.

By analogy with the 1D case, the $d$D diagonal preconditioner also has diagonal entries of the form $2^{-s}$ where $s$ represents a scale parameter.
The basis functions in $d$D are the tensor product of 1D basis functions, so the scale parameter in $d$D is determined by the largest scale parameter of 1D basis functions in the tensor product.
That is, if $s_1,s_2,\ldots,s_d$ are the scale parameters of each 1D basis function, then the scale parameter in $d$D is $s_{\max}=\max(s_1,\ldots,s_d)$.
Therefore, diagonal entries of the $d$D preconditioner have the form $2^{-s_{\max}}$.

The action of the $d$D preconditioner~$P_{d\text{D}}$ on basis states also follows by analogy with the action of the 1D preconditioner~$P$.
For the 1D case, the action on the $n$-qubit basis state can be written as $P\ket{j}=2^{-\floor{\log_2j}}\ket{j}$ with no action on $\ket{0}$.
Note that the floor function appears here because $j\in\{0,\ldots,2^n-1\}$ whereas $s\in\{0,\ldots,n-1\}$.
Similarly, the action of the $d$D preconditioner on the $d$D basis states $\ket{j_1}\ket{j_2}\cdots\ket{j_d}$,
the tensor product of 1D basis states, is
$P_{d\textup{D}} \ket{j_1}\ket{j_2}\cdots\ket{j_d}
= 2^{-\floor{\log_2 j_{\max}}} \ket{j_1}\ket{j_2}\cdots\ket{j_d}$,
where $j_{\max}:=\max(j_1,\ldots,j_d)$.
If $j_{\max}=0$, no action is performed.

\subsection{Implementing the multidimensional preconditioner}

The action of the 1D preconditioner $P$ on the $n$-qubit basis state $\ket{j}$ is $P\ket{j} = 2^{-\floor{\log_2j}}\ket{j}$ for all $j\in\{1,\ldots,2^n-1\}$;
$P$ has no effect on~$\ket{0}$.
For preconditioning, we described in the main text that we only need to implement
\begin{equation}
\label{eq:U1D}
    U^\pm \ket{j}=\text{e}^{\pm\text{i}\theta_j}\ket{j}
\end{equation}
with $\cos\theta_j=2^{-\floor{\log_2j}}$.
Notice that here we only need to apply $n-1$ rotations for the $n$-qubit state $j$.
For $d>1$, the action of $d$D preconditioner $P_{d\text{D}}$ on the state $\ket{j_1}\ket{j_2}\cdots\ket{j_d}$, the tensor product of $d$ $n$-qubit basis states, is
\begin{equation}
\label{eq:dDprecond}
    P_{d\text{D}} \ket{j_1}\ket{j_2}\cdots\ket{j_d}
    = 2^{-\floor{\log_2 j_{\max}}} \ket{j_1}\ket{j_2}\cdots\ket{j_d} \quad \forall j_{\max}\neq 0,
\end{equation}
where $j_{\max} = \max(j_1,\ldots,j_d)$.
The action is trivial for $j_{\max}=0$;
note that $j_{\max}\in\{0,\ldots,2^n-1\}$.
Similar to the 1D case, the $d$D preconditioner can be decomposed as a linear combination of two unitary operations.
Specifically, $P_{d\text{D}}$ can be decomposed~as
\begin{equation}
    P_{d\text{D}}=\frac{1}{2}\left(U^+_{d\text{D}}+U^{-}_{d\text{D}}\right),
\end{equation}
where 
\begin{equation}
    U^\pm_{d\text{D}}:=
    P_{d\text{D}}\pm\text{i}\sqrt{\mathbbm{1}-P^2_{d\text{D}}}
    =\text{e}^{\pm\text{i}\arccos{P_{d\text{D}}}}
\end{equation}
are unitary operations.
Analogous to the 1D case, the action of $U^\pm_{d\text{D}}$ on basis states is obtained from the preconditioner's action on basis states.
By Eq.~\eqref{eq:dDprecond}, the action of $U^\pm_{d\text{D}}$ on basis states is
\begin{equation}
    U^\pm_{d\text{D}} \ket{j_1}\ket{j_2}\cdots\ket{j_d}
    = \text{e}^{\pm\text{i}\theta_{\max}} \ket{j_1}\ket{j_2}\cdots\ket{j_d},
\end{equation}
where $\cos{\theta_{\max}} = 2^{-\floor{\log_2 j_{\max}}}$.
A simple approach to implement $U^\pm_{d\text{D}}$ is using the similarity between the action of $U^\pm_{d\text{D}}$ and $U^\pm$ on basis states.
Note that if we compute $j_{\max}$ into an $n$-qubit ancilla register, then the transformation in Eq.~\eqref{eq:UdD} would be similar to the transformation in Eq.~\eqref{eq:U1D}.
The quantum circuit in Fig.~\ref{fig:qCircs}(c) shows an implementation of~$U^\pm_{d\text{D}}$ using this simple approach.

\section{Detailed description and complexity analysis}
\label{apx:detailedComplexity}

\subsection{Block encoding for the preconditioned matrix}
\label{apx:BEncode_Ap}

Here we construct a zero-error block encoding for the preconditioned matrix $A_P$ using a single call to the zero-error block encoding of $A$ and a number of gates that scales as $\mathcal{O}(n^2)$.
Our construction follows from the fact that $A_P$ is obtained from~$A$ by a unitary transformation~$W$ and a diagonal scaling~$P$, i.e., the relation $A_P=PWAW^\dagger P$ between $A_P$ and $A$.
First we construct a zero-error block-encoding for the preconditioner $P$ and then use it to construct a block-encoding for~$A_P$.

We use the Linear Combination of Unitaries~(LCU) method~\cite{CW12} to construct a zero-error block-encoding for~$P$.
Specifically, by the decomposition $P=\frac{1}{2}(U^++ U^-)$ with $U^\pm$ defined in Eq.~\eqref{eq:Upm}, we obtain the $(1,1,0)$-block-encoding
\begin{equation}
\label{eq:U_P}
    U_P: = (H\otimes\mathbbm{1}_n)\Lambda_0(U^+)\Lambda_1(U^-)(H\otimes\mathbbm{1}_n)
\end{equation}
for $P$, where $H$ is the Hadamard gate and $\Lambda_b(U)$ is the $\ket{b}$-controlled-$U$ gate with $b\in\{0,1\}$.
The gate cost of implementing $U_P$ is $\mathcal{O}(n)$, which is the gate cost of implementing controlled-$U^\pm$ given in Lemma~\ref{lemma:cUpm}.
Now given a $(1,a,0)$-block-encoding $U_A$ for~$A$, a $(1,a,0)$-block-encoding $U_{A_P}$ for $A_P$ is constructed as
\begin{equation}
    U_{A_P} :=
    (\mathbbm{1}_{a-1} \otimes U_P)
    (\mathbbm{1}_a\otimes W)
    U_A
    (\mathbbm{1}_a\otimes W^\dagger)
    (\mathbbm{1}_{a-1} \otimes U_P),
\end{equation}
where we used the relation $A_P=PWAW^\dagger P$ and the fact that $W$ is unitary.
This block-encoding makes a single call to $U_A$ and its gate cost is dominated by the gate cost of $W$, which is $\mathcal{O}(n^2)$.

\subsection{Block encoding for the inverse of the preconditioned matrix}
\label{apx:qmi}
We now give a simple procedure for implementing the block-encoding operation \textsc{qmi} in Fig.~\ref{fig:qCircs}(a) and analyze its computational cost.
The approach we describe here is only to elucidate the action of \textsc{qmi} by a simple procedure and not to focus on its efficiency.
Finally, we state the cost of advanced methods for executing \textsc{qmi} which provide exponentially better scaling with respect to~$\varepsilon$.

For convenience, we state the action of \textsc{qmi} in Eq.~\eqref{eq:qmi} as follows.
For any $n$-qubit input state~$\ket{\psi_\text{in}}$, this operation performs the transformation
\begin{equation}
\label{eq:apx_qmi}
    \textsc{qmi}:
    \ket{0}_\texttt{flag}\ket{\psi_\text{in}}
    \mapsto
    \ket{0}_\texttt{flag} \frac{1}{\alpha}
    \widetilde{A}^{-1}_P\ket{\psi_\text{in}}
    + \ket{\phi},
\end{equation}
where $A_P$ is the preconditioned matrix, $\widetilde{A}^{-1}_P$ is an $\varepsilon$-approximation of $A^{-1}_P$ in the operator norm
and $\ket{\phi}$ is an unnormalized $(n+1)$-qubit state such that $({}_\texttt{flag}\!\bra{0}\otimes\mathbbm{1}_n)\ket{\phi}=0$.
For simplicity of discussion, here we take $\alpha=\norm{A_P^{-1}}$ but it could have a value greater than~$\norm{A_P^{-1}}$.
The flag qubit in Eq.~\eqref{eq:apx_qmi} is one of the $a$ ancilla qubits
in Eq.~\eqref{eq:qmi}, which we later use to mark success.
In general, $a\geq 1$ because $A_P$ is a nonunitary matrix.
The rest of the ancilla qubits are not displayed here for simplicity.

First we note that the eigenvalues of the preconditioned matrix $A_P$ are in the range $[1/\kappa_p,1]$, where $\kappa_p$ is the condition number of~$A_P$.
This is because (i)~the largest eigenvalue of $A$ is at most one by the assumption that $\norm{A}\leq 1$;
(ii)~the largest eigenvalue of preconditioner $P$ is one (see Fig.~\ref{fig:precond}(a));
and (iii)~$A$ is a positive-definite matrix by condition~(\ref{cond3}).
Therefore, the largest~eigenvalue of $A_P=PA_WP$ is less than or equal to one, so its eigenvalues are in $[1/\kappa_p,1]$.
Consequently, eigenvalues of $A^{-1}_P$ are in $[1,\kappa_p]$.

Let $\lambda_\ell \in [1/\kappa_p,1]$ and $\ket{v_\ell}$ be the eigenvalues and eigenvectors of $A_P$, respectively, and let
$\ket{\psi_\text{in}}=\sum_\ell \beta_\ell \ket{v_\ell}$ be the decomposition of the input state in eigenbasis of $A_P$.
Then applying the quantum phase-estimation algorithm, denoted by \textsc{qpe}, with unitary $U:=\exp{(\text{i}2\pi A_P)}$ on the input state yields
\begin{equation}
    \textsc{qpe} \ket{0^t}\ket{\psi_\text{in}}
    = \sum_\ell \beta_\ell \ket{\widetilde{\lambda}_\ell}\ket{v_\ell},
\end{equation}
where $\widetilde{\lambda}_\ell$ is an approximation of the eigenvalue such that $|\lambda_\ell-\tilde{\lambda}_\ell|\leq \varepsilon_\textsc{qpe}$
with an error $\varepsilon_\textsc{qpe}> 0$, and
\begin{equation}
\label{eq:t}
    t\in \mathcal{O}(\log(1/\varepsilon_\textsc{qpe}))
\end{equation}
is the number of bits in binary representation of eigenvalues.
To produce an approximation of the state $A^{-1}_P\ket{\psi_\text{in}}=\sum_\ell \beta_\ell/\lambda_\ell\ket{v_\ell}$ on the right-hand side, we use the controlled-rotation operation
\begin{equation}
    \textsc{crot}
    \ket{0}_\texttt{flag}
    \ket{\tilde{\lambda}_\ell}
    = 
    \left(
    \frac{1}{\alpha\widetilde{\lambda}_\ell}\ket{0}_\texttt{flag}
    +\sqrt{1-\frac{1}{(\alpha\widetilde{\lambda}_\ell)^2}}\ket{1}_\texttt{flag}
    \right)
    \ket{\tilde{\lambda}_\ell}
\end{equation}
that rotates the single qubit controlled on the value encoded in the second register.
By \textsc{qpe} and \textsc{crot}, the mapping \textsc{qmi} in Eq.~\eqref{eq:apx_qmi} can be implemented as
\begin{align}
    \ket{0}_\texttt{flag}
    \ket{0^t}
    \ket{\psi_\text{in}}
    &\xrightarrow{\mathbbm{1}_1\otimes\,\textsc{qpe}}
    \sum_\ell \beta_\ell
    \ket{0}_\texttt{flag}
    \ket{\widetilde{\lambda}_\ell}
    \ket{v_\ell}\\
    &\xrightarrow{\textsc{crot}\,\otimes\mathbbm{1}_n}
    \sum_\ell \beta_\ell
    \left(\frac{1}{\alpha\widetilde{\lambda}_\ell}
    \ket{0}_\texttt{flag}+\sqrt{1-\frac{1}{(\alpha\widetilde{\lambda}_\ell)^2}}
    \ket{1}_\texttt{flag}\right)
    \ket{\widetilde{\lambda}_\ell}
    \ket{v_\ell}\\
    &\xrightarrow{\mathbbm{1}_1\otimes\,\textsc{qpe}^\dagger}
    \sum_\ell \beta_\ell 
    \left(\frac{1}{\alpha\widetilde{\lambda}_\ell}\ket{0}_\texttt{flag}
    +\sqrt{1-\frac{1}{(\alpha\widetilde{\lambda}_\ell)^2}}
    \ket{1}_\texttt{flag}\right)
    \ket{0^t}
    \ket{v_\ell}
    \approx \frac{1}{\alpha}
    \ket{0}_\texttt{flag}
    \ket{0^t}
    A_P^{-1}\ket{\psi_\text{in}}
    + \ket{\phi'},
\end{align}
where $\ket{\phi'}$ is a state similar to $\ket{\phi}$ in Eq.~\eqref{eq:apx_qmi} but with extra $t$ qubits in the all-zero state $\ket{0^t}$.
The described implementation can be summarized as 
$\textsc{hhl} := (\mathbbm{1}_1\otimes\textsc{qpe}^\dagger) (\textsc{crot}\otimes\mathbbm{1}_n) (\mathbbm{1}_1\otimes\textsc{qpe})$
which is a high-level description of the HHL algorithm~\cite{HHL09} without amplitude amplification.

The computational cost of the described approach for implementing \textsc{qmi}
is determined by the cost of \textsc{crot} and \textsc{qpe}. 
The \textsc{crot} operation can be implemented using $t$ single-qubit controlled rotations~\cite[Fig.~13]{BSB+22}, so its gate cost is $\mathcal{O}(\log(1/\varepsilon_\textsc{qpe}))$ using Eq.~\eqref{eq:t}.
The query cost of \textsc{qpe} is $\mathcal{O}(2^t)$, where each query is to perform the controlled-$U$ operation with $U:=\exp{(\text{i}2\pi A_P)}$.
Therefore, by Eq.~\eqref{eq:t}, the query cost for implementing \textsc{qpe} is $\mathcal{O}(1/\varepsilon_\textsc{qpe})$.
This is indeed the number of calls to the block-encoding of $\exp{(\text{i}2\pi A_P)}$, which is a constant-time Hamiltonian simulation.
The block-encoding of $\exp{(\text{i}2\pi A_P)}$ with error $\varepsilon'$ can be constructed using $\mathcal{O}(\log(1/\varepsilon'))$ calls to the block-encoding of~$A_P$~\cite[Algorithm~4]{MRT+21}.
To ensure that the error in implementing \textsc{qpe} is upper bounded by $\varepsilon_\textsc{qpe}$, 
we need $2^t \varepsilon'\leq \varepsilon_\textsc{qpe}$, so $\varepsilon' \in \mathcal{O} (\varepsilon^2_\textsc{qpe})$ by Eq.~\eqref{eq:t}.
Therefore, the query cost for implementing \textsc{qpe} in terms of calls to the block-encoding of $A_P$ is $\mathcal{O}((1/\varepsilon_\textsc{qpe})\log(1/\varepsilon_\textsc{qpe}))$.
This is also the query cost in terms of calls to the block-encoding $U_A$ of $A$ because the block-encoding of $A_P$ can be constructed by one call to $U_A$ as shown in Appendix~\ref{apx:BEncode_Ap}.

The error for phase estimation $\varepsilon_\textsc{qpe}$ is determined from the smallest eigenvalue $\lambda_{\min}$ of $A_P$ and the upper bound $\varepsilon$ for the error in generating the desired state.
Specifically, in order to generate an approximation of
the desired state with an error bounded by~$\varepsilon$, the error $\varepsilon_\textsc{qpe}$ needs to be bounded as
$\varepsilon_\textsc{qpe} \leq  \lambda_{\min}\varepsilon/2 =  \varepsilon/(2\kappa_p)$.
By this bound, the second inequality in 
\begin{equation}
    \frac{1}{\alpha}
    \norm{
    \sum_\ell \beta_\ell \left(\frac{1}{\lambda_\ell} - \frac{1}{\widetilde{\lambda}_\ell}\right)
    \ket{v_\ell}}
    \leq 
    \frac{1}{\alpha}
    \norm{
    \sum_\ell \frac{\beta_\ell}{\lambda_\ell} 
    \frac{|\lambda_\ell-\widetilde{\lambda}_\ell|}{\widetilde{\lambda}_\ell}
    \ket{v_\ell}}
    \leq 
    \frac{\varepsilon}{\alpha}
    \norm{\sum_\ell\frac{\beta_\ell}{\lambda_\ell} \ket{v_\ell}}
    =\frac{\varepsilon}{\alpha}
    \norm{A^{-1} \ket{\psi_\text{in}}}
    \leq \varepsilon
\end{equation}
holds and the overall error between the true and approximate states is bounded by~$\varepsilon$.
Therefore, by $\varepsilon_\textsc{qpe}\in\mathcal{O}(\varepsilon/2\kappa_p)$ and the query and gate costs of \textsc{crot} and \textsc{qpe}, the query and gate cost for implementing \textsc{qmi} are
$\mathcal{O}((\kappa_p/\varepsilon)\log(\kappa_p/\varepsilon))
= \mathcal{O}((1/\varepsilon)\log(1/\varepsilon))$
and $\mathcal{O}(\log(\kappa_p/\varepsilon))=\mathcal{O}(\log(1/\varepsilon))$, respectively,
where the identity follows from $k_p \in \mathcal{O}(1)$.

The described approach provides a simple procedure for implementing the block-encoding \textsc{qmi} of $A^{-1}_P$ in Fig.~\ref{fig:qCircs}(a).
Using advanced methods based on the linear combination of unitaries~\cite{CKS17}, quantum singular-value transformation~\cite{GSL+19,MRT+21}, quantum eigenstate filtering~\cite{LT20} and discrete adiabatic theorem~\cite{CAS+22}, the query cost for constructing the block-encoding \textsc{qmi} is exponentially improved with respect to $\varepsilon$.
Specifically, the query cost of these methods is $\mathcal{O}(k_p\log(k_p/\varepsilon))$ in terms of calls to the block-encoding of $A_P$;
the query cost in Ref.~\cite{CAS+22} is indeed $\mathcal{O}(k_p\log(1/\varepsilon))$.
In our application $k_p \in \mathcal{O}(1)$, so we have 
\begin{equation}
\label{eq:Gqmi}
    \mathcal{Q}_\textsc{qmi}
    \in \mathcal{O}(\log(1/\varepsilon))
\end{equation}
for the query cost of constructing the block-encoding \textsc{qmi}.
This is also the query cost in terms of calls to the block-encoding $U_A$ of $A$ because the block-encoding of $A_P$ can be constructed by one call to $U_A$; see Appendix~\ref{apx:BEncode_Ap}. 

The query complexity of the advanced methods in Refs.~\cite{CKS17,GSL+19,MRT+21} is indeed determined from the degree of an odd polynomial approximating the rescaled inverse function $1/(\kappa_px)$ on the domain $I:=[-1,-k_p]\cup [k_p,1]$.
Let $\delta:=1/k_p$, then the degree of the odd polynomial $p(x)$ that is $\varepsilon'$-close to $\delta/x$ on the domain $I$ can be chosen to be
$\mathcal{O}(\frac{1}{\delta}\log\frac{1}{\varepsilon'})$
by~\cite[Corollary 69]{GSL+19};
an explicit construction for~$p(x)$ is given in Ref.~\cite{CKS17}.
To guarantee that $ \tilde{A}^{-1}_P= p(A_P^{-1})/\delta$ is $\varepsilon$-close to $A^{-1}_P$ in the operator norm, we need $\varepsilon'\in \mathcal{O}(\varepsilon/\delta)$.
Therefore, the query cost is
$\mathcal{O}(\frac{1}{\delta}\log\frac{1}{\varepsilon'})
=
\mathcal{O}(\kappa_p\log(\kappa_p/\varepsilon)).
$

\subsection{Amplitude amplification in our application}
\label{apx:qaa}

Here we show that only $\mathcal{O}(1)$ rounds of amplitude amplifications are needed on average to have a high success probability for the solution state in our application.
We provide a detailed procedure for each round of amplitude amplification and prove Lemma~\ref{lemma:AA} by analyzing the overall computational cost of amplitude amplification.

We begin by showing the expected number of amplitude amplifications in our application is $\mathcal{O}(1)$.
Let
\begin{equation}
\label{eq:psiout}
\ket{\psi_\text{out}}:=
    \frac{1} {\sqrt{2}}
    \sum_{a\in\{0,1\}}\ket{a}\ket{\psi_a},
    \quad  \ket{\psi_a}:= U^a W\ket{\bm{b}},
\end{equation}
be the state of the last two registers in Fig.~\ref{fig:qCircs}(a) after applying the controlled-$U^\pm$ operations and before \textsc{qmi}.
Then by the action of \textsc{qmi} in Eq.~\eqref{eq:apx_qmi} with $\ket{\psi_\text{in}}=\ket{\psi_a}$,
the input state to $\mathcal{A}$ in Fig.~\ref{fig:qCircs}(a) can be written as
\begin{equation}
\label{eq:Psi}
    \ket{\Psi}:=
    \frac{1}{\alpha\sqrt{2}}
    \sum_{a\in\{0,1\}}
    \ket{0}_\texttt{flag}
    \ket{a} A^{-1}_P\ket{\psi_a}
    + \ket{\Phi},
\end{equation}
where $\ket{\Phi}$ is an $(n+2)$-qubit unnormalized state with $({}_\texttt{flag}\!\bra{0}\otimes\mathbbm{1}_{n+1})\ket{\Phi}=0$.
Specifically, the input state $\ket{\Psi}$ can be generated using \textsc{qmi} as
\begin{equation}
\label{eq:genPsi}
    \ket{\Psi}=(\textsc{swap}\otimes\mathbbm{1}_n)(\mathbbm{1}_1\otimes\textsc{qmi})(\textsc{swap}\otimes\mathbbm{1}_n)\ket{0}_\texttt{flag}\ket{\psi_\text{out}},
\end{equation}
where \textsc{swap} is the two-qubit swap gate and $\ket{\psi_\text{out}}$ is given in Eq.~\eqref{eq:psiout}.
Let us now define the normalized state
\begin{equation}
    \ket{\text{G}}:=\frac{1}{\mathcal{N}}
    (\mathbbm{1}_1\otimes A^{-1}_P)\ket{\psi_\text{out}}
    =\frac{1}{\mathcal{N}\sqrt{2}}
    \sum_a \ket{a}A^{-1}_P\ket{\psi_a}
\end{equation}
with the normalization factor $\mathcal{N}:=\norm{(\mathbbm{1}_1\otimes A^{-1}_P)\ket{\psi_\text{out}}}$.
Then the state $\ket{\Psi}$ in Eq.~\eqref{eq:Psi} can be decomposed as
\begin{equation}
    \ket{\Psi}=
    \sqrt{p}\ket{0}_\texttt{flag}\ket{\text{G}}
    + \sqrt{1-p}\ket{\text{B}},
\end{equation}
where $\ket{0}_\texttt{flag}\!\ket{\text{G}}$ is the `good' part of $\ket{\Psi}$ with the success probability
$p:=(\mathcal{N}/\alpha)^2$ and the normalized state $\ket{\text{B}}=\ket{\Phi}/\sqrt{1-p}$ is the `bad' part of $\ket{\Psi}$.
We now establish a lower bound for $p$.
Note that $\mathcal{N}^2\geq 1$ because
\begin{equation}
    \mathcal{N}^2 = \frac{1}{2}
    \norm{\sum_{a\in\{0,1\}}\ket{a}A^{-1}_p\ket{\psi_a}}^2
    = \frac{1}{2}\sum_a\norm{A^{-1}_P\ket{\psi_a}}^2
    \geq 
    \frac{1}{2}
    \sum_a \min_{\ket{\psi_a}}
    \norm{A^{-1}_P\ket{\psi_a}}^2
    = \frac{1}{2}
    \sum_a \left(\lambda_{\min} \left(A^{-1}_P\right)\right)^2
    \geq 1,
\end{equation}
where the last equality follows from the fact that eigenvalues of $A^{-1}_P$ lie in the range $[1,\kappa_p]$.
Therefore, a lower bound for $p$ is
\begin{equation}
    p=\frac{\mathcal{N}^2}{\alpha^2}
    \geq \frac{1}{\alpha^2} =
    \frac{1}{\kappa^2_p},
\end{equation}
so $p\in\mathcal{O}(1)$ because
$\kappa_p \in \mathcal{O}(1)$ in our application.
The value of $p$ is constant but unknown.
Nonetheless, by~\cite[Theorem 3]{BHM+02}, we can boost the success probability to say~$p=2/3$ using an expected rounds of amplitude amplifications that is in $\mathcal{O}(1)$.

We now describe a procedure for implementing each round of amplitude amplification.
Let $R_n:=2\ketbra{0^n}{0^n}-\mathbbm{1}_n$ be the $n$-qubit reflection operator with respect to the $n$-qubit zero state $\ket{0^n}$.
Each round of amplitude amplification is composed of two reflection operations:
one about the good state and the other about the initial state.
Specifically, let
\begin{align}
    \label{eq:RGood}
    R_\text{Good}&:=
    R_0 \otimes \mathbbm{1}_{n+1}
    = Z \otimes \mathbbm{1}_{n+1}, \\
    R_\text{Initial}&:= 2\ketbra{\Psi}{\Psi}-\mathbbm{1}_{n+2}
    = U_\Psi R_{n+2} U_\Psi^\dagger
\end{align}
be the reflections about the good and initial states, respectively.
Then the amplitude amplification is
\begin{equation}
\label{eq:AAprocedure}
    \mathcal{A}:=
    R_\text{Initial} R_\text{Good}
    = (U_\Psi R_{n+2} U_\Psi^\dagger)(Z\otimes\mathbbm{1}_{n+1}),
\end{equation}
where $U_{\Psi}$ is the unitary that prepares the $(n+2)$-qubit initial state $\ket{\Psi}$ from the all-zero state, i.e., $U_{\Psi}\ket{0^{n+2}}=\ket{\Psi}$.
Notice here that the good state $\ket{0}_\texttt{flag}\!\ket{\text{G}}$ is marked by the flag qubit, so the reflection about the good state is constructed by the reflection about $\ket{0}$ of the flag qubit as in Eq.~\eqref{eq:RGood}.

By Eq.~\eqref{eq:AAprocedure}, we need an implementation for the $(n+2)$-qubit reflection $R_{n+2}$ and an implementation for $U_{\Psi}$ to construct a procedure for implementing each round of amplitude amplification $\mathcal{A}$.
Using Eq.~\eqref{eq:genPsi}, the unitary $U_{\Psi}$ is composed of \textsc{qmi}, \textsc{swap} and the operations that prepare the state $\ket{\psi_\text{out}}$ in Eq.~\eqref{eq:psiout}.
By the circuit in Fig.~\ref{fig:qCircs}(a), we have
\begin{equation}
    \ket{\psi_\text{out}}=\frac{1} {\sqrt{2}}
    \sum_{a\in\{0,1\}}\ket{a}\ket{\psi_a}
    = \Lambda_1(U^-) \Lambda_0(U^+)
    (\mathbbm{1}_1\otimes W)
    (H\otimes\mathcal{P}_{\bm{b}})\ket{0}\ket{0^n},
\end{equation}
where $\Lambda_b(U)$ with $b\in\{0,1\}$ is the $\ket{b}$-controlled-$U$ operation.
Therefore,
\begin{equation}
\label{eq:Upsi}
    U_{\Psi} =
    (\textsc{swap}\otimes\mathbbm{1}_n)(\mathbbm{1}_1\otimes\textsc{qmi})(\textsc{swap}\otimes\mathbbm{1}_n)
    \Lambda_1(U^-) \Lambda_0(U^+)
    (\mathbbm{1}_1\otimes W)
    (H\otimes\mathcal{P}_{\bm{b}}),
\end{equation}
which yields an implementation for $U_{\Psi}$.
The reflection operator $R_{n+2}$ can be implemented using phase kickback and one ancilla qubit.
Specifically, the identity
\begin{equation}
\label{eq:Rn}
    (R_n\otimes \mathbbm{1}_1)\ket{\psi}\ket{01}
    = -X^{n+1}(\mathbbm{1}_n\otimes H)\Lambda_1^n(X)
    (\mathbbm{1}_n\otimes H)X^{n+1}\ket{\psi}\ket{01},
\end{equation}
yields an implementation for $R_n$ up to an irrelevant global $-1$ phase factor, 
where $\ket{\psi}$ is any $n$-qubit state and
\begin{equation}
   \Lambda_1^n(X):=\ketbra{1^n}{1^n}\otimes X + (\mathbbm{1}_n-\ketbra{1^n}{1^n})\otimes\mathbbm{1}_1  
\end{equation}
is the $(n+1)$-bit Toffoli gate.
Having described a procedure for implementing $\mathcal{A}$, we now prove Lemma~\ref{lemma:AA} by analyzing the overall computational cost of amplitude amplification.

\begin{proof}[Proof of Lemma~\ref{lemma:AA}]
By Eqs.~\eqref{eq:AAprocedure} and~\eqref{eq:Upsi}, the gate cost for each round of amplitude amplification is determined by the cost of executing the $(n+2)$-qubit reflection $R_{n+2}$, \textsc{qmi}, controlled-$U^\pm$ and $W$.
Two uses of the procedure $\mathcal{P}_{\bm{b}}$ is also needed.
By Eq.~\eqref{eq:Rn}, the reflection $R_{n+2}$ can be performed using one ancilla qubit and one $(n+3)$-bit Toffoli gate, which can be implemented using $\mathcal{O}(n)$ elementary one- and two-qubit gates~\cite[Corollary~7.4]{BBC+95}.
The gate cost for executing the quantum wavelet transform $W$ is $\mathcal{O}(n^2)$~\cite{BA24,FW99}.

By the query cost given in Eq.~\eqref{eq:Gqmi} for executing \textsc{qmi} and the gate cost given in Lemma~\ref{lemma:cUpm} for executing controlled-$U^\pm$, and because only $\mathcal{O}(1)$ rounds of amplitude amplification on average is needed in our application, executing amplitude amplification needs
$\mathcal{O}(1)$ uses of $\mathcal{P}_{\bm{b}}$, $\mathcal{O}(\log(1/\varepsilon))$ uses of the block-encoding $U_A$ of~$A$ and $\mathcal{O}(n^2)$ gates, all on average.
\end{proof}

\subsection{Recovering the norm of the solution state by repetition}
\label{apx:norm}

The solution state that our algorithm generates is a normalized state stored on a quantum register.
However, the normalization factor is needed for computing the expectation value of a given observable; see Eq.~\eqref{eq:expMprime}.
The normalization factor can be obtained from the probability of success state, i.e., the state $\ket{0}$, for the flag qubit by repeating the algorithm without amplitude amplification.
The success probability is
\begin{equation}
    p_\text{succ}
    =\frac{1}{4\alpha^2}
    \norm{\sum_{ab}\ket{ab}\ket{\psi_{ab}}}^2
    =\frac{1}{4\alpha^2}
    \sum_{ab} \norm{\ket{\psi_{ab}}}^2
    =\frac{\xi^2}{\kappa^2_p},
\end{equation}
where the last identity follows from Eq.~\eqref{eq:psi}.
Hence, sufficiently repeating the algorithm and estimating $p_\text{succ}$ yields the normalization factor $\xi = \kappa_p \sqrt{p_\text{succ}}$ for the solution state $\ket{\psi}$ in Eq.~\eqref{eq:psi}, which is needed for computing the expectation value in~Eq.~\eqref{eq:expMprime}.

\subsection{The controlled unitaries}
\label{apx:cU}

In this appendix, we give a procedure to implement the controlled-$U^\pm$ with the $n$-qubit unitary $U^\pm$ given in Eq.~\eqref{eq:Udecomp}.
We show that controlled-$U^\pm$ can be executed using $n$ ancilla qubits, $\mathcal{O}(n)$ Toffoli gates, $\mathcal{O}(n)$ Pauli-$X$ gates, and $\mathcal{O}(n)$ controlled-rotation gates, thereby proving Lemma~\ref{lemma:cUpm}.

By Fig.~\ref{fig:qCircs}(b), $U^\pm$ is a product of multi-controlled rotations, where all controls are on the state $\ket{0}$ of control qubits.
Therefore, controlled-$U^\pm$ is also a product of multi-controlled rotations but with one extra control qubit.
The extra control for $U^+$ is on the state $\ket{0}$ and for $U^-$ is on the state $\ket{1}$; see Fig.~\ref{fig:qCircs}(a).
Each $\ket{0}$-control can be made $\ket{1}$-control by applying a Pauli-$X$ gate before and after the control, so only $\mathcal{O}(n)$ Pauli-$X$ gates are needed to make all controls to be on the $\ket{1}$ of qubits.

Therefore, we need to implement a product of multi-controlled rotations of the form
$\Lambda_1^1(R_z(\theta_0))
\Lambda_1^2(R_z(\theta_1))
\cdots
\Lambda_1^n(R_z(\theta_{n-1}))$
where $\Lambda_1^r(R_z(\theta))$,
for some angle $\theta$, is the
$\ket{1^r}$-controlled-$R_z(\theta)$ operation.
A simple approach to implement such a product using $n-1$ ancilla qubits is as follows.
Let $\ket{c_1,c_2,\ldots,c_n}$ be the state of the control qubits in the computational basis.
First compute $c_1c_2$ into the first ancilla qubit using a Toffoli gate; the state of the first ancilla is transformed as $\ket{0}\mapsto\ket{c_1c_2}$.
Then compute $c_1c_2c_3$ into the second ancilla qubit using a Toffoli gate, followed by computing $c_1c_2c_3c_4$ into the third ancilla qubit using a Toffoli gate.
By continuing this process, the state of $n-1$ ancilla qubits transforms as $\ket{0,0,\ldots,0}\mapsto\ket{c_1c_2, c_1c_2c_3,c_1c_2c_3c_4,\ldots,c_1c_2c_3\cdots c_n}$.
Now to implement $\Lambda_1^r(R_z(\theta))$ with $r\geq 2$, apply $R_z(\theta)$ on the target qubit controlled on the state $\ket{1}$ of the $(r-1)$th ancilla~qubit.
Finally, uncompute the ancilla qubits by applying Toffoli gates.
Evidently, the described procedure needs $\mathcal{O}(n)$ Toffoli gates and $\mathcal{O}(n)$ single-controlled rotations. 

\subsection{Computing the minimum or maximum of values encoded in quantum registers}
\label{apx:minmax}

In this appendix, we describe a procedure for implementing the \textsc{max} operation defined
as
\begin{equation}
\label{eq:appMax}
    \textsc{max} \ket{j_1}_\texttt{reg$1$}
    \ket{j_2}_\texttt{reg$2$}
    \cdots\ket{j_d}_\texttt{reg$d$}
    \ket{0^n}_\texttt{out}
    := \ket{j_1}_\texttt{reg$1$}
    \ket{j_2}_\texttt{reg$2$}
    \cdots\ket{j_d}_\texttt{reg$d$}
    \ket{j_{\max}}_\texttt{out},
\end{equation}
where $j_{\max}:=\max(j_1,\cdots,j_d)$ and each register comprises $n$ qubits encoding an $n$-bit number.
We also describe how the procedure for \textsc{max} can be modified to implement \textsc{min} operation where $j_{\max}$ is replaced with $j_{\min}:=\min(j_1,\cdots,j_d)$.
Finally, we discuss the computational cost of our implementations for \textsc{max} and \textsc{min}, and prove Lemma~\ref{lemma:max}.

In our implementations, we use a quantum comparator operation \textsc{comp} defined as
\begin{equation}
    \textsc{comp}\ket{x}_\texttt{reg$i$}
    \ket{y}_\texttt{reg$j$}
    \ket{0}_\texttt{flag}:=
    \begin{cases}
    \ket{x}_\texttt{reg$i$}
    \ket{y}_\texttt{reg$j$}
    \ket{1}_\texttt{flag} \;\;\text{if}\, x\geq y, \\
    \ket{x}_\texttt{reg$i$}
    \ket{y}_\texttt{reg$j$}
    \ket{0}_\texttt{flag} \;\;\text{if}\, x<y,
    \end{cases}
\end{equation}
where the flag qubit is flipped if the value encoded in the first register is greater than or equal to the value encoded in the second register;
note that the flag qubit here is not the same as the flag qubit used in the main text.
Let us define \textsc{cadd} operation as
\begin{equation}
\label{eq:cadd}
    \textsc{cadd} \ket{x}_{\texttt{reg}i}
    \ket{y}_{\texttt{reg}j}
    \ket{f}_\texttt{flag}
    \ket{0^n}_\texttt{out}
    := \begin{cases}
    \ket{x}_{\texttt{reg}i}
    \ket{y}_{\texttt{reg}j}
    \ket{f}_\texttt{flag}
    \ket{x}_\texttt{out}
    \;\;\text{if}\, f=1, \\
    \ket{x}_{\texttt{reg}i}
    \ket{y}_{\texttt{reg}j}
    \ket{f}_\texttt{flag}
    \ket{y}_\texttt{out} 
    \;\,\,\text{if}\, f=0,
    \end{cases}
\end{equation}
which adds the value encoded in the first or second register to the last register controlled by the value of the flag register: $x$ is added if \texttt{flag} is $\ket{1}$ and $y$ is added if \texttt{flag} is $\ket{0}$.
The value encoded in \texttt{out} register is $\max(x,y)$.
To compute $\min(x,y)$ into \texttt{out}, we only need to modify \textsc{cadd} so that it adds $y$ to \texttt{out} if \texttt{flag} is $\ket{1}$ and adds $x$ if \texttt{flag} is $\ket{0}$.

Using \textsc{comp} and \textsc{cadd}, \textsc{max} can be performed recursively in $\mathcal{O}(\log_2 d)$ steps as follows.
In the first step, compute the maximum of values encoded in each consecutive pair \texttt{reg2$j$} and \texttt{reg2$j$+1} and write the result into an $n$-qubit temporary register \texttt{tmp$j$}.
For each pair, \textsc{comp} and \textsc{cadd} are performed once and one flag qubit is used.
In the second step, compute the maximum of the values encoded in each consecutive pair \texttt{tmp2$j$} and \texttt{tmp2$j$+1} and write the result into a new temporary register.
These operations are repeated in each step, but the number of pairs in each step is reduced by two.
In the last step, $j_{\max}$ is computed into \texttt{out} register, and all temporary and flag qubits are erased by appropriate uncomputations.

The number of temporary registers in the first step is $d/2$; the second step is $d/4$; the third step is $d/8$, and so forth.
Hence the total number of temporary registers needed in the described approach is $d(1/2+1/4+1/8+\cdots)\leq d$.
Similarly, the number of needed flag qubits is at most $d$.
As each temporary register comprises $n$ qubits, the total number of ancilla qubits needed to implement \textsc{max} in Eq.~\eqref{eq:appMax} is $\mathcal{O}(dn)$.
The gate cost is obtained similarly.
Specifically, the gate cost $\mathcal{G}_\textsc{max}$ to implement \textsc{max} is
\begin{equation}
\label{eq:Gmax}
    \mathcal{G}_\textsc{max} = d(\mathcal{G}_\textsc{comp}
    +\mathcal{G}_\textsc{cadd})(1/2+1/4+\cdots) \times 2 \leq
    2d(\mathcal{G}_\textsc{comp}
    +\mathcal{G}_\textsc{cadd}),
\end{equation}
where the extra factor of $2$ comes from the cost of uncomputation.
The \textsc{cadd} operation in Eq.~\eqref{eq:cadd} can be implemented using $2n$ Toffoli gates.
As shown in Ref.~\cite{Gid18}, the quantum comparator \textsc{comp} can be executed using $n$ Toffoli gates and $n$ ancilla qubits.
Therefore, $\mathcal{G}_\textsc{max}\in \mathcal{O}(dn)$ by Eq.~\eqref{eq:Gmax} yielding Lemma~\ref{lemma:max}.
We also have $\mathcal{G}_\textsc{min}\in \mathcal{O}(dn)$ by the above discussion.

\subsection{Block-encoding approach for linear systems with uniform condition number}
\label{apx:qlsa}

Having constructed a unitary block-encoding for the preconditioned matrix~$A_P$ in Appendix~\ref{apx:BEncode_Ap}, the QSVT approach~\cite{GSL+19} can be used to generate a quantum solution for the preconditioned linear system $A_P\ket{\bm{u}_P}=\ket{\bm{b}_P}$,
where $A_P$ has a bounded condition number as~$\kappa_p\leq c$ for~$c$ a constant number independent of the size of the matrix~$A_P$.
Here we show that the bounded condition number enables constructing a simpler polynomial approximation for the inverse function compared to the general case given in Lemmas~17--19 of Ref.~\cite{CKS17}.

To solve a linear system by QSVT, we need to find an odd polynomial that $\varepsilon$-approximates the inverse function~$\textsc{inv}(x):=1/x$ over the range of singular values~$\sigma_\ell$ of~$A_P$, which belong to the interval $I_c:=[1/c,1]$ by the discussion in Appendix~\ref{apx:qmi}.
As required by QSVT, the polynomial must be bounded in magnitude by~$1$.
So we seek a polynomial approximation to the rescaled inverse function~$(1/2c)1/x$ on~$I_c$.
Notice that this function has magnitude~$\leq 1/2$ because of the prefactor~$1/2$; this prefactor is only used for later simplification~\cite[p.~24]{MRT+21}.
The output of the QSVT is then an approximation of~$(1/2c)A^{-1}_P$.
Therefore, due to the multiplicative factor~$1/2c$, we seek an~$\varepsilon/2c$ approximation to the rescaled function, which yields an~$\varepsilon$ approximation to~$\textsc{inv}$.

We provide the appropriate polynomial in the following lemma.

\begin{lemma}
Let~$I^+_c:=[1/c,1]$ and~$I^-_c:=(-1,-1/c)$ for~$c\geq 1$.
Also, for $\varepsilon\in(0,1)$, let~$P_{\varepsilon,I^\pm_c}^\textsc{inv}(x)$ be a polynomial~that is $\varepsilon$-close to the inverse function $\textsc{inv}(x):=1/x$ on the interval~$I^\pm_c$ and let~$P_{\varepsilon/c,I^\pm_c}^\textsc{step}(x)$ be a polynomial that is $(\varepsilon/c)$-close to the unit step function~$\textsc{step}(x)$, which is~$1$ for $x>0$ and~$0$ otherwise. Then the odd polynomial
\begin{equation}
\label{eq:oddPoly}
    P_{\varepsilon/c,I_c}^\textsc{mi}(x) =
     (1/2c)P_{\varepsilon,I^+_c}^\textsc{inv}(x)
     P^\textsc{step}_{\varepsilon/c,I^+_c}(x)
    - (1/2c)P_{\varepsilon,I^-_c}^\textsc{inv}(-x)
    P^\textsc{step}_{\varepsilon/c,I^-_c}(-x)
\end{equation}
is $(\varepsilon/c)$-close to~$1/(2cx)$ on~$I_c:=I^+_c\cup I^-_c$ and its magnitude is bounded from above by~$1$.
\end{lemma}
\begin{proof}
Without loss of generality, we consider the positive subinterval $I^+_c$ where $x>0$. In this case, we have
\begin{align}
    \abs{P_{\varepsilon/c,I_c}^\textsc{mi}(x)-\frac{1}{2cx}}
    &= \frac{1}{2c}\abs{P_{\varepsilon,I^+_c}^\textsc{inv}(x)
     P^\textsc{step}_{\varepsilon/c,I^+_c}(x)-\frac{1}{x}}\\
    & \leq \frac{1}{2c}
    \left(
    \abs{P_{\varepsilon,I^+_c}^\textsc{inv}(x)-\frac{1}{x}}
    \abs{P^\textsc{step}_{\varepsilon/c,I^+_c}(x)}
    +\abs{\frac{1}{x}}\abs{P^\textsc{step}_{\varepsilon/c,I^+_c}(x)-1}
    \right)\\
    & \leq \frac{1}{2c}\left(\varepsilon\times 1 + c \times \frac{\varepsilon}{c}\right) = \frac{\varepsilon}{c},
\end{align}
where the first inequality is obtained by adding and subtracting $(1/x)P^\textsc{step}_{\varepsilon/c,I^+_c}(x)$ and using the triangle inequality.
We now bound the magnitude of $P_{\varepsilon/c,I_c}^\textsc{mi}(x)$ for $x>0$ as
\begin{equation}
    \abs{P_{\varepsilon,I_c}^\textsc{mi}(x)} = \frac{1}{2c} \abs{P_{\varepsilon,I^+_c}^\textsc{inv}(x)} \abs{ P^\textsc{step}_{\varepsilon/c,I^+_c}(x)} \leq \frac{1}{2c}(c+\varepsilon) \times 1 \leq 1,
\end{equation}
where we simply used the bound on the magnitudes of \textsc{inv} and \textsc{step} on $I^+_c$.
The same bounds hold for $I^-_c$.
\end{proof}

We now show that the bounded condition number enables constructing a simple polynomial approximation for \textsc{inv}.
To this end, let us consider the case $\kappa_p\leq c=32$, where the constant value~$32$ is chosen for illustration.
We want an $\varepsilon/2\kappa_p$ approximation to $(1/\kappa_p)(1/2x)$ in the interval~$[-1,1]\setminus [-1/c,1/c]$.
Without loss of generality, we consider the positive part of this interval.
Let us define
$z_0:=1+2/c$ and
$x':=2x-z_0$;
notice $z_0>1$ for any $c$ and $x'\in [-1,1]$.
First, we find a polynomial approximation for $1/2x=1/(x'+z_0)$.
Using Eqs.~(8)--(10) in Ref.~\cite{mat06}, for any $x\in[-1,1]$ and $z\notin [-1,1]$, we have
\begin{equation}
    \frac{1}{x+z} = \frac{1}{\sqrt{z^2-1}} + \frac{2}{\sqrt{z^2-1}}\sum_{\ell=1}^{\infty} \frac{(-)^\ell}{(z+\sqrt{z^2-1})^\ell} T_\ell(x),
\end{equation}
where~$T_\ell(x)$ are the Chebyshev polynomials of the first kind.
We use this equation to approximate~$1/(x'+z_0)$ by polynomials.
To this end, we note that $z_0+\sqrt{z_0^2-1}\geq \sqrt{2}$,
so we need to find a sufficient upper bound $\ell_{\max}$ for the sum over~$\ell$ that gives an $\varepsilon/2$ approximation to $1/2x$.
The sufficient $\ell_{\max}$ is in $\Omega(\log{(1/\varepsilon)})$, because we need to find an $\ell_{\max}$ such that
$(1/\sqrt{2})^{\ell_{\max}}\leq \varepsilon/2$,
which yields $\ell_{\max} \geq 2+ 2\log(1/\varepsilon)$.

Having a polynomial approximation for \textsc{inv}, we then use $P^\textsc{step}_{\varepsilon,I^\pm_\kappa}(\pm x)=1/2\pm(1/2)P^\textsc{sign}_{\varepsilon,I_\kappa}(x)$ and a polynomial approximation for \textsc{sign} as in Ref.~\cite{MRT+21} to obtain a polynomial approximation for \textsc{step}.
The polynomial approximations for \textsc{inv} and \textsc{step} yield the desired polynomial approximation for matrix inversion by Eq.~\eqref{eq:oddPoly}.
\end{document}